\pdfoutput=1
\RequirePackage{rotating}
\RequirePackage{latexml}
\iflatexml
	\documentclass{article}
	\author{% 
	Paul Gölz \\ UC Berkeley, Cornell University
	\and Dominik Peters \\ CNRS, LAMSADE, Université Paris Dauphine - PSL
	\and Ariel D. Procaccia \\ School of Engineering and Applied Sciences, Harvard University}
\else
	\documentclass[11pt]{scrartcl}
	\usepackage[a4paper, total={16cm, 24cm}]{geometry}
	\setcapindent{0pt}
	\setkomafont{captionlabel}{\sffamily\bfseries}

	\usepackage{authblk}

	\author[1]{Paul Gölz}
	\author[2]{Dominik Peters}
	\author[3]{Ariel D.\ Procaccia}
	\affil[1]{UC Berkeley, Cornell University}
	\affil[2]{CNRS, LAMSADE, Université Paris Dauphine - PSL}
	\affil[3]{School of Engineering and Applied Sciences, Harvard University}
\fi

\usepackage{booktabs} %
\usepackage[ruled]{algorithm2e} %

\SetAlFnt{\small}
\SetAlCapFnt{\small}
\SetAlCapNameFnt{\small}
\SetAlCapHSkip{0pt}
\IncMargin{-\parindent}

\usepackage{natbib}

\usepackage[pagebackref]{hyperref}
\hypersetup{
	pdfencoding=auto, 
	psdextra,
	colorlinks=true,
	citecolor=green!50!black,
	linkcolor=red!60!black,
	urlcolor=blue!90!black
}

\usepackage{amsmath,amsfonts,amsthm}
\usepackage{xcolor}
\definecolor{brewer1}{HTML}{8dd3c7}
\definecolor{brewer2}{HTML}{ffffb3}
\definecolor{brewer3}{HTML}{bebada}
\definecolor{brewer4}{HTML}{fb8072}
\usepackage{xfrac}
\usepackage{tikz}
\usetikzlibrary{shapes,calc,fit,decorations.pathreplacing}
\usepackage{mathtools,thmtools}
\usepackage[nameinlink,capitalize]{cleveref}
\iflatexml\else
\usepackage{crossreftools}
\fi

\usepackage{todonotes}
\usepackage{comment}
\usepackage{dsfont}
\usepackage{bm}
\usepackage{enumerate}

\allowdisplaybreaks

\newcommand{\emdash}{\,---\,}

\newcommand{\pvec}[1]{\vec{#1}\mkern2mu\vphantom{#1}}
\newcommand{\bvn}{Birkhoff--von Neumann}
\newcommand{\onebar}[1]{\overbracket[0.5pt][0pt]{#1}{}\mkern-2.5mu}
\newcommand{\wideBar}[1]{\overbracket[0.5pt][0pt]{\overbracket[0.5pt][0pt]{#1}}}
\newcommand{\twobar}[1]{\wideBar{#1}{}\mkern-2.5mu}
\newcommand{\cut}{\mathit{cut}}
\newcommand{\degr}{\mathit{deg}}
\newcommand{\frdeg}{\mathit{frac}}

\newcommand{\qone}{\lfloor q_i(t\!-\!1) \rfloor}
\newcommand{\qtwo}{\lfloor q_i(t) \rfloor}
\newcommand{\aone}{a_i(t\!-\!1)}
\newcommand{\atwo}{a_i(t)}
\newcommand{\patrue}{$\checkmark$ $\atwo = \aone + 1$}
\newcommand{\pafalse}{\hphantom{$\checkmark$} $\atwo = \aone$}
\newcommand{\pbtrue}{$\checkmark$ $\aone = \qone + 1$}
\newcommand{\pbfalse}{\hphantom{$\checkmark$} $\aone = \qone$}
\newcommand{\pctrue}{$\checkmark$ $\atwo = \qtwo$}
\newcommand{\pcfalse}{\hphantom{$\checkmark$} $\atwo = \qtwo + 1$}

\newcommand{\natsone}{\mathbb{N}}
\newcommand{\natszero}{\mathbb{Z}_{\geq 0}}

\newtheorem{theorem}{Theorem}
\crefname{theorem}{Theorem}{Theorems}
\newtheorem{lemma}{Lemma}
\crefname{lemma}{Lemma}{Lemmas}
\theoremstyle{plain}
\newtheorem{construction}[theorem]{Construction}
\crefname{construction}{Construction}{Constructions}
\newcommand{\bone}{\mathds{1}}

\iflatexml
\title{In This Apportionment Lottery,\\ the House Always Wins}
\else
\title{\vspace{-1cm}In This Apportionment Lottery,\\ the House Always Wins}
\fi

\definecolor{RED}{HTML}{FF0000}

\date{\vspace{-12pt}\sffamily \large Version v2 as accepted by \emph{Operations Research} \\ \medskip June 2024\vspace{-16pt}}

\begin{document}
\maketitle

\begin{abstract}
\small
\emph{Apportionment} is the problem of distributing $h$ indivisible seats across states in proportion to the states' populations. In the context of the US House of Representatives, this problem has a rich history and is a prime example of interactions between mathematical analysis and political practice.
\citet{Grim04} suggested to apportion seats in a \emph{randomized} way such that each state receives exactly their proportional share $q_i$ of seats in expectation (\emph{ex ante proportionality}) and receives either $\lfloor q_i \rfloor$ or $\lceil q_i \rceil$ many seats ex post (\emph{quota}).
However, there is a vast space of randomized apportionment methods satisfying these two axioms, and so we additionally consider prominent axioms from the apportionment literature.
Our main result is a randomized method satisfying quota, ex ante proportionality and \emph{house monotonicity}\emdash{}a property that prevents paradoxes when the number of seats changes and which we require to hold ex post. 
This result is based on a generalization of dependent rounding on bipartite graphs, which we call \emph{cumulative rounding} and which might be of independent interest, as we demonstrate via applications beyond apportionment.
\end{abstract}

\vspace{-0.8cm}
\enlargethispage{0.3cm}
\renewcommand*{\contentsname}{}
{\small
\tableofcontents
}

\section{Introduction}
The Constitution of the United States says that 
\begingroup
\addtolength\leftmargini{-2.65mm}
\begin{quote}
``Representatives [in the US House of Representatives] shall be apportioned among the several States according to their respective numbers, counting the whole number of persons in each State \dots''
\end{quote} 
\endgroup
These ``respective numbers,'' or populations, of the states are determined every decade through the census. For example, on April 1, 2020, the population of the United States was 331,108,434, and the state of New York had a population of 20,215,751.
New York therefore deserves 6.105\% of the 435 seats in the House, which is 26.56 seats, for the next ten years.
The puzzle of apportionment is what to do about New York's 0.56 seat\emdash{}in this round of apportionment it was rounded down to 0, and New York lost its 27th seat.

The mathematical question of how to allocate these fractional seats has riveted the American political establishment since the country's founding~\citep{Szpir10}.
In 1792, Congress approved a bill that would enact an apportionment method proposed by Alexander Hamilton, the first secretary of the treasury and star of the eponymous musical.
If we denote the \emph{standard quota} of state $i$ by $q_i$ ($q_i=26.56$ in the case of New York in 2020), Hamilton's method allocates to each state its \emph{lower quota} $\lfloor q_i \rfloor$ (26 for NY). %
Then, Hamilton's method goes through the states in order of decreasing \emph{residue} $q_i-\lfloor q_i\rfloor$ (0.56 for NY) and allocates an additional seat to each state until all house seats are allocated.

As sensible as Hamilton's method appears, it repeatedly led to bizarre results, which became known as \emph{apportionment paradoxes}.
\begin{description}%
\item[The Alabama paradox:] Using the 1880 census results, the chief clerk of the Census Office calculated the apportionment according to Hamilton's method for all House sizes between 275 and 350, and discovered that, as the size increased from 299 to 300, Alabama lost a seat. %
In 1900, the Alabama paradox reappeared, this time affecting Colorado and Maine. 

\item[The population paradox:] In 1900, the populations of Virginia and Maine were 1,854,184 and 694,466, respectively. Over the following year, the populations of the two states grew by 19,767 and 4,649, respectively. Even though Virginia's growth was larger even relative to its population, Hamilton's method would have transferred a seat from Virginia to Maine. 
\end{description}
Occurrences of these paradoxes caused partisan strife, which is only natural since a state's representatives have a strong personal stake in their state not losing seats.
Both in Congress and the courts, this strife took the form of a tug-of-war over the choice of apportionment method, the size of the House, and the census numbers, driven by the states', parties', and individual representatives' self-interest rather than the public good.

This state of affairs improved in 1941 when Congress adopted an apportionment method that provably avoids the Alabama and population paradoxes, which had been developed by Edward Huntington, a Harvard mathematician, and Joseph Hill, the chief statistician of the Census Bureau.
While the Huntington--Hill method is \emph{house monotone} (i.e., it avoids the Alabama paradox) and \emph{population monotone} (i.e., it avoids the population paradox), it has a different, equally bizarre weakness:
it does not \emph{satisfy quota}, that is, the allocation of some states may be different from $\lfloor q_i \rfloor$ or $\lceil q_i\rceil$.
A striking impossibility result by \citet{BY82} shows that this tension is inevitable: no apportionment method can simultaneously satisfy quota and be population monotone. (We will revisit this result in \cref{sec:popmon} and show that, while Balinski and Young's theorem makes additional implicit assumptions, the incompatibility between quota and population monotonicity continues to hold without these assumptions.)

While the Balinski--Young impossibility is troubling, in our view there is an even larger source of unfairness that plagues apportionment methods, which is rooted in their determinism.
In addition to introducing bias (the Huntington-Hill method disadvantages larger states), deterministic methods often lead to situations where small counting errors can change the outcome.
For example, based on the 2020 census, New York lost its 27th House seat, but it would have kept it had its population count been higher by 89 residents!
Indeed, current projections suggest that New York would have kept its seat were it not for distortions in census response rates~\citep[p.~20]{EMS+21}.
After the 1990 and 2000 censuses, similar circumstances were the basis for lawsuits brought by Massachusetts and Utah.
A second shortcoming of deterministic apportionment methods is a lack of fairness over time: For example, if the states' populations remain static, a state with a standard quota of, say, $1.5$ might receive a single seat in every single apportionment and therefore only receive $2/3$ of its deserved representation. 

To address these issues, an obvious solution is to use randomization in order to realize the standard quota of each state in expectation, as Grimmett proposed in 2004~\citep{Grim04}.
If such a randomized method was used, 89 additional residents would have shifted New York's expected number of seats by a negligible $0.0001$, and the decision between 26 or 27 seats would have been made by an impartial random process, which is less accessible to political maneuvering than, say, the census~\citep{Stone11}.

Grimmett's proposed apportionment method is easy to describe. First, it chooses a random permutation of the states; without loss of generality, that permutation is identity.
Second, it draws $U$ uniformly at random from $[0,1]$, and let $Q_i \coloneqq U + \sum_{j=1}^i q_j$.
Finally, it allocates to each state $i$ one seat for each integer contained in the interval $[Q_{i-1},Q_i)$.
In particular, this implies that the allocation will satisfy quota.

Why this particular method? \citet[p.~302]{Grim04} writes:
\begingroup
\addtolength\leftmargini{-2.65mm}
\begin{quote}
``We offer no justification for this scheme apart from fairness and ease of implementation.''
\end{quote}
\endgroup
Grimmett's method is easy to implement for sure, and what he refers to as ``fairness''\emdash{}realizing the fractional quotas in expectation\emdash{}is arguably a minimal requirement for any randomized apportionment method.
But his two axioms, ``fairness'' and quota, allow for a vast number of randomized methods:
Indeed, after allocating $\lfloor q_i \rfloor$ seats to each agent, the problem of determining which states to round up reduces to so-called ``$\pi$ps sampling'' (``inclusion probability proportional to size''), and dozens of such schemes have been proposed in the literature~\citep{BH83}.
We believe, therefore, that additional criteria are needed to guide the design of randomized apportionment methods. To identify such criteria, we return to the classics: house and population monotonicity.

\subsection{Our Approach and Results}
In this paper, we seek randomized apportionment methods that satisfy natural extensions of house and population monotonicity to the randomized setting.
We want these monotonicity axioms to hold even \emph{ex post}, i.e., after the randomization has been realized.
We find such methods by taking a parameterized class of deterministic methods all of which satisfy the desired ex post axioms (in our case, subsets of population monotonicity, house monotonicity, and quota), and to then randomize over the choice of parameters such that ex ante properties hold (here: ex ante proportionality).
In mechanism design, a similar approach extends strategyproofness to universal strategyproofness~\citep{NR01}.

Guaranteeing monotonicity axioms ex post is helpful for preventing certain kinds of manipulation in the apportionment process.
For instance, say that the census concludes and a randomized apportionment is determined, and only then does a state credibly contest that its population was undercounted (in the courts or in Congress with the support of a majority).
Using an apportionment method without population monotonicity, states might strategically undercount their population in the census and only reveal the true count in case this turns out to be beneficial once the randomness is revealed.
When using a population monotone method, by contrast, any revised apportionment would be made using the same deterministic and population monotone method, which implies that immediately revealing the full population count is a dominant strategy.

In \cref{sec:popmon}, we first show that no such randomized methods exist for population monotonicity.
This impossibility is not due to randomization or ex ante proportionality, but due to the fact that population monotonicity and quota are outright incompatible. Thus, there do not exist suitable deterministic apportionment methods that a randomized apportionment method could randomize over.
That population monotonicity and quota are incompatible is well-known from the Balinski--Young impossibility theorem~\citep{BY82}.
But their proof uses some ``mild'' background conditions 
that are not mild for our randomized purposes since, to provide ex ante proportionality, the randomized method must sometimes prioritize smaller states over larger states with positive probability, which is ruled out by those background conditions.
We are able to prove a stronger version of their theorem, which derives the impossibility with no assumptions other than population monotonicity and quota.
The deterministic apportionment methods that are most commonly used in practice (so called divisor methods, which include the Huntington--Hill method) satisfy population monotonicity but fail quota. So it makes sense to ask whether population monotonicity can be combined with ex ante proportionality (without requiring quota).
We construct such a method, which is reminiscent of the family of divisor methods, except that the so-called ``divisior criterion''~\citep{BY82} is specific to each state and is given by a sequence of Poisson arrivals.

For house monotonicity, we provide in \cref{sec:housemon} a randomized apportionment method that satisfies house monotonicity, quota, and ex ante proportionality.
To obtain this result, we generalize the classic result of \citet{GKP+06} on dependent rounding in a bipartite graph.
We call this method \emph{cumulative dependent randomized rounding} or just \emph{cumulative rounding}~(\cref{thm:cumulative}).
Cumulative rounding allows to correlate dependent-rounding processes in multiple copies of the same bipartite graph such that the result satisfies an additional guarantee across copies of the graph. This guarantee, which we describe in the next paragraph, generalizes the quota axiom of apportionment.
As a side product, our existence proof for house monotonicity provides a new characterization of the deterministic apportionment methods satisfying house monotonicity and quota, which is based on the corner points of a bipartite matching polytope.

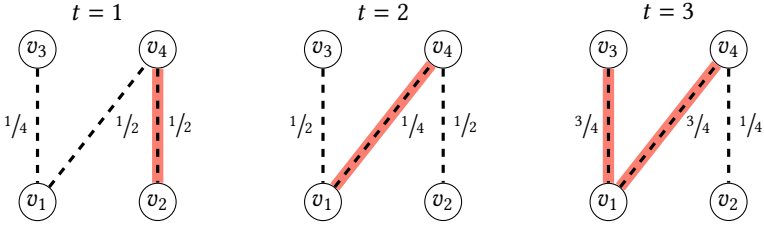
\begin{figure}
    \centering
    \begin{tikzpicture}[node distance=1.5cm]
    \tikzstyle{vertex}=[draw,circle,inner sep=0.5mm];
        \tikzstyle{pot}=[dashed,line width=1.2];

    \foreach \num/\numone/\denone/\numtwo/\dentwo/\numthree/\denthree/\onone/\ontwo/\onthree in {1/1/4/1/2/1/2/0/0/1,2/1/2/1/4/1/2/0/1/0,3/3/4/3/4/1/4/1/1/0} {
    \node at (\num*3.8, 0) (t\num) {$t = \num$};
    
     \node [vertex,xshift=-.8cm,yshift=-.5cm] at (t\num) (v\num{}3) {$v_3$};;
     \node [vertex,xshift=.8cm,yshift=-.5cm] at (t\num) (v\num{}4) {$v_4$};
     \node [vertex,below=of v\num{}3] (v\num{}1) {$v_1$};
     \node [vertex,below=of v\num{}4] (v\num{}2) {$v_2$};

     \if\onone1
     \draw [line width=1.5mm,brewer4](v\num{}1)-- (v\num{}3);
     \fi
     \if\ontwo1
     \draw [line width=1.5mm,brewer4](v\num{}1)-- (v\num{}4);
     \fi
     \if\onthree1
     \draw [line width=1.5mm,brewer4](v\num{}2)-- (v\num{}4);
     \fi

      \draw (v\num{}3) [pot] -- node[left] {$\sfrac{\numone}{\denone}$} (v\num{}1) [pot] --node[right,xshift=1mm] {$\sfrac{\numtwo}{\dentwo}$} (v\num{}4) [pot]--node[right] {$\sfrac{\numthree}{\denthree}$} (v\num{}2); 
    };
\end{tikzpicture}
    \caption{Illustration of cumulative rounding. Dashed lines indicate edges $e \in E$ in the bipartite graph $(V,E)$, which are labeled with weights $w_e^t$. The red lines indicate a possible random outcome of cumulative rounding.}
    \label{fig:cumulativerounding}
\end{figure}
To describe cumulative rounding more precisely, we first sketch the result of \citet{GKP+06}. For a bipartite graph $(V, E)$ and edge weights $\{w_e\}_{e \in E}$ in $[0,1]$, dependent rounding randomly generates a subgraph $(V, E')$ with $E' \subseteq E$ providing three properties: \emph{marginal distribution} (each edge $e \in E$ is contained in $E'$ with probability $w_e$), \emph{degree preservation} (in the rounded graph, the degree of a vertex $v$ is the floor or the ceiling of $v$'s fractional degree $\sum_{v \in e \in E} w_e$), and \emph{negative correlation}.
Cumulative rounding allows us to randomly round $T$ many copies of $(V, E)$, where each copy $1 \leq t \leq T$ has a set of weights $\{w_e^t\}_{e \in E}$. Each copy will provide marginal distribution, degree preservation, and negative correlation.
As we prove in \cref{sec:cumulativeproof}, cumulative rounding additionally guarantees \emph{cumulative degree preservation}: for each vertex $v$ and $1 \leq t \leq T$, the sum of degrees of $v$ across copies $1$ through $t$ equals the sum of fractional degrees of $v$ across copies $1$ through $t$, either rounded up or down.
For example, node $v_1$ in \cref{fig:cumulativerounding} is incident to edges with a total fractional weight of $2 \cdot \sfrac{1}{4} + 2 \cdot \sfrac{1}{2} = 1.5$ across copies $t=1,2$, and must hence be incident to $1$ or $2$ edges in total across the rounded versions of copies $t = 1,2$.
Since, across copies $t=1,2,3$, $v_1$'s total fractional degree is $1.5 + 2 \cdot \sfrac{3}{4} = 3$, $v_1$ must be incident to a total of exactly $3$ rounded edges across the copies $t=1,2,3$.
By applying cumulative rounding to a star graph, we obtain a randomized apportionment method satisfying house monotonicity, quota, and ex-ante proportionality.

We believe that cumulative rounding is of broader interest, and in \cref{sec:applications}, we present applications of cumulative rounding beyond apportionment. 
First, we consider a proposal of \citet{BH09} for reforming the European Commission of the European Union: They propose a weighted lottery to determine which countries nominate commissioners.
Using cumulative rounding to implement this lottery would eliminate two key problems the authors identified in a simulation study, in particular
the possibility that some member states may not nominate any commissioners for a long time.
We also describe how to apply cumulative rounding to round fractional allocations of goods or chores, and we discuss a specific application of assigning faculty to teach courses.

\subsection{Related Work}
\paragraph{Apportionment theory.}
We have already mentioned several seminal works of apportionment theory.
Besides the axiomatic approach~\cite[e.g.,][]{BY75,BY82,BR99,BR14,PPR24}, which has arguably proved the most influential, deterministic apportionment has been extensively studied through the lens of constrained optimization~\cite[e.g.,][]{Huntington28,BH63,Ernst94,Agnew08}, with respect to bias against larger or smaller states~\cite[e.g.,][]{Polya19,BY82,MOP02,LV06,Janson14}, and in multi-dimensional generalizations~\cite[e.g.,][]{BD89a,BD89,MZZ10,CCV22,MV22}.
For a comprehensive treatment of this theory, we refer the reader to \citet{BY82} and \citet{Pukelsheim17}.

\paragraph{Randomized apportionment.}
A na\"ive form of randomized apportionment was suggested by \citet{Balinski93}, who immediately rejected it: ``It is trivial to propose an unbiased method: assign the $h$ seats at random with probabilities proportional to the fair shares. In this case none of the other desirable properties is guaranteed.''
The proposal by \citet{Grim04}, which we discussed above, makes a much stronger case for randomized apportionment by showing one desirable property\emdash{}quota\emdash{}which can, in fact, be guaranteed ex post.
Our work adds house and population monotonicity to the set of achievable properties.

\citet{ALM+19} developed a random rounding scheme as part of a mechanism for strategy-proof peer selection, which they simultaneously proposed as a randomized apportionment method.
Like Grimmett's method, their method satisfies ex ante proportionality and quota.
The main advantage of their method is that its support consists of only linearly (not exponentially) many deterministic apportionments.
This, \citeauthor{ALM+19} argued, is useful in repeated apportionment settings, where one could repeat a periodic sequence of these deterministic apportionments and thereby limit the possibility of selecting the same state much too frequently or much too rarely due to random fluctuations.
If this is the goal, cumulative rounding will arguably give better guarantees (see \cref{sec:europeancommission}).

\citet{HNR+23} proposed pipage rounding~\citep{GKP+06}\,---\,in this case equivalent to pivotal sampling~\citep{DT98}\,---\,as a randomized apportionment method, without pursuing monotonicity.

\citet{CCS+24} proposed a randomized apportionment scheme that circumvents the impossibility from \cref{sec:popmonnegative} by allowing the house size to deviate ex post from its target. Their scheme satisfies ex-ante proportionality, quota, and population monotonicity, along with probabilistic bounds on how far the house size may deviate.
\citet{CCS+24} also provided a conceptually simpler version of our characterization of house monotone and quota-compliant apportionment solutions in \cref{prop:quotatone}.

\citet{EK24} studied a problem closely related to ours, but inspired by affirmative action for faculty hiring in Indian universities.
One can think of a house monotone and quota-compliant apportionment method as an iterative process that, in each time step $t=1, 2, \dots$, allocates the $t$th house seat to one of the states, ensuring that the total number of seats awarded to each state is proportional up to rounding.
In the same way, \citeauthor{EK24} successively allocate a university department's vacancies to demographic groups, and they also aim for quota and ex-ante proportionality.
Their algorithm is essentially equivalent to our randomized apportionment method; instead of randomly rounding a matching, they round a flow that appears similar to the one of \citet{CCS+24}.
By independently randomizing over hiring decisions in each of several departments, the authors immediately obtain concentration bounds, which imply that, in total across departments, demographic groups are likely almost proportionally represented.

\citet{CGS+24} studied randomized apportionment to target\emdash{}in addition to ex ante proportionality and quota\emdash{}monotonicity axioms that are quite different from ours.
The general flavor of these monotonicity axioms is to require that, if the standard quotas of some set $S$ of states weakly increase and the standard quotas of all other states weakly decrease, the probability that the states in $S$ simultaneously receive more seats should weakly increase, i.e., they impose monotonicity on higher-order correlations in the rounding.
None of the apportionment methods they consider satisfy house monotonicity \citep[App.\ A]{CGS+24}.

\paragraph{Fair division.}
Apportionment can be seen as a special case of the fair division of indivisible goods, which has recently received increased attention in operations research~\citep[e.g.,][]{SS22,AFS+23,BKP+23}.
The apportionment setting is characterized by the fact that the goods (i.e., the seats) are interchangeable and that the agents (i.e., the states) are \emph{weighted} (in other words, have different entitlements)~\citep[e.g.,][]{Barbanel96,AMS20}.
Though the focus on interchangeable goods may appear very restrictive at first glance, \cite{CSS21} showed that house monotone apportionment methods, when interpreted as \emph{picking sequences}, induce allocation algorithms for the full setting of weighted fair division.
In particular, an apportionment method satisfying house monotonicity and quota yields a fair-division algorithm satisfying \emph{weighted proportionality up to one good (WPROP1)}~\citep[Prop.\ 4.8]{CSS21}.
Hence, our randomized method in \cref{sec:housemon} can be seen as a randomized picking-sequence algorithm that ensures WPROP1 ex post, and ex ante satisfies that each agent $i$ will get the $t$th pick (for each $t = 1, 2, \dots$) with probability proportional to $i$'s weight.
This latter property is not only an intuitively appealing fairness guarantee in its own right, but also immediately implies \emph{weighted proportionality}~\citep{Barbanel96}.

Our work is part of a larger thrust to develop allocation mechanisms that combine desirable ex-ante and ex-post guarantees, which have been termed \emph{best-of-both-worlds guarantees}~\citep{AFS+23}.
The works of \cite{HZ79} and \cite{BM01} were early precursors to this idea, in the setting of matchings where the space of realizable ex-ante probabilities has a very clean structure~\citep{Birk46,Neumann53}.
\cite{BCKM13} generalized this approach to more general combinatorial constraints.
\cite{AFS+23} first studied classic fair-division axioms in this way, and \cite{BEF22} and \cite{FMN+23} extended this approach to additional fairness axioms and more general valuations.
The corollary in the previous paragraph is a best-of-both-worlds fairness guarantee for weighted fair division; recently, \cite{AGM23} and \cite{HSV23} have obtained such guarantees for this setting.
\paragraph{Randomly rounding bipartite matchings.}
As a consequence of the \bvn{} Theorem~\citep{Birk46,Neumann53}, any fractional matching in a bipartite graph can be implemented as a lottery over integral matchings, in the sense that each edge is present in the random matching with probability equal to its weight in the fractional matching.
One algorithm for rounding a bipartite matching is pipage rounding~\citep{AS04}, which \citet{GKP+06} randomized in their dependent rounding technique. 
This rounding technique is powerful since it can directly accommodate fractional degrees larger than 1
and can provide negative-correlation properties such that Chernoff concentration bounds apply~\citep{PS97}.
The technique of \citeauthor{GKP+06} has found many applications in approximation algorithms~\citep{GKP+06,KMP+09,BGLM+12} 
and in fair division~\citep{SS18a,AN20,CJMW19}.

\paragraph{Just-in-time production.}
\citet{SY93} studied a problem in just-in-time industrial manufacturing: how to alternate between the production of different types of goods in a way that produces each type in specified proportions.
As pointed out by \citet{BCC96} and expanded upon by \citet{BS98}, this problem is related to apportionment. In particular, a production schedule resembles a deterministic, house monotone apportionment method: as the available production time increases by one slot, the schedule needs to decide which type to produce in the next slot.
\citeauthor{SY93} ended up with a property that nearly guarantees quota because they aim to minimize how far the prevalence of types among the goods produced so far deviates from the desired proportions. 
(Most of the literature on this just-in-time production problem minimized other measures of deviation~\citep[see][]{Kubiak93}, which are not connected to quota.)
Now, \citet{SY93} only produce deterministic schedules, and the existence of deterministic house monotone and quota apportionment methods has long been known~\citep{BY75,Still79}. But we believe that the main construction in their proof could be randomized to obtain an alternative proof of \cref{thm:housemono}, without however providing the generality of cumulative rounding. In fact, a similar graph construction to that by \citeauthor{SY93} is randomly rounded within a proof by \citet{GKP+06} to obtain an approximation result about broadcast scheduling.

\section{Model}
\label{sec:model}
Throughout this paper, fix a set of $n \geq 2$ states $N = \{1, 2, \dots, n\}$.
For a given \emph{population profile} $\vec{p} \in \natsone^n$, which assigns a \emph{population} of $p_i \in \natsone$ to each state $i$, and for a \emph{house size} $h \in \natsone$, an apportionment solution deterministically allocates to each state $i$ a number $a_i \in \natszero$ of house seats such that the total number of allocated seats is $h$.
Formally, a \emph{solution} is a function $f : \natsone^n \times \natsone \to \natszero^n$ such that, for all $\vec{p}$ and $h$, $\sum_{i \in N} f_i(\vec{p}, h) = h$.
For a population profile $\vec{p}$ and house size $h$, state $i$'s \emph{standard quota} is $q_i \coloneqq \frac{p_i}{\sum_{i \in N} p_i} \, h$.

Next, we define three axioms for solutions:
\begin{description}%
\item[Quota:] A solution $f$ satisfies quota if, for any $\vec{p}$ and $h$, it holds that $f_i(\vec{p}, h) \in \{\lfloor q_i \rfloor, \lceil q_i \rceil\}$ for all states $i$.
\item[House monotonicity:] A solution $f$ satisfies house monotonicity if, for any $\vec{p}$ and $h$, increasing the house size to $h+1$ does not reduce any state's seat number, i.e., if $f_i(\vec{p}, h) \leq f_i(\vec{p}, h+1)$ for all $i \in N$.
\item[Population monotonicity:] We say that a solution $f$, some $\vec{p}, \pvec{p}' \in \natsone^n$, and some $h, h' \in \natsone$ exhibit a \emph{population paradox} if there are two states $i \neq j$ such that $p_i' \geq p_i$, $p_j' \leq p_j$, $f_i(\pvec{p}', h') < f_i(\vec{p}, h)$, and $f_j(\pvec{p}', h') > f_j(\vec{p}, h)$, or, in words, if state $i$ loses seats and $j$ wins seats even though $i$'s population weakly grew and $j$'s population weakly shrunk.
A solution $f$ is \emph{population monotone} if it exhibits no population paradoxes for any $\vec{p}, \pvec{p}', h, h'$.
By setting $\vec{p} = \pvec{p}'$, one easily verifies that population monotonicity implies house monotonicity.
\end{description}

Note that the apportionment literature often considers two components of the quota axiom separately: \emph{lower quota} (``$f_i(\vec{p},h) \geq \lfloor q_i \rfloor$'') and \emph{upper quota} (``$f_i(\vec{p},h) \leq \lceil q_i \rceil$'').
Also note that our definition of population monotonicity, taken from \citet{RU10}, is slightly weaker than the definition of other authors, whose violation we describe in the introduction.
All results extend to this alternative notion of \emph{relative population mo\-no\-to\-ni\-ci\-ty}~\citep{RU10}: the proof of \cref{prop:quotapopmon} immediately applies, and the proof of  \cref{thm:proppopmon} is easy to adapt.
Our results also continue to apply if one weakens population monotonicity by requiring that $h' = h$.

Finally, we will define randomized apportionment methods.
One potential definition, used by \citet{Grim04}, is a function that for each $\vec{p}$ and $h$ specifies a probability distribution over seat allocations $(a_i)_{i\in N}$. 
For us, an apportionment method is instead a random process that determines an entire (deterministic) apportionment solution, i.e., apportionments for all population profiles $\vec{p}$ and house sizes $h$.
The advantage of this definition is that it allows us to formulate axioms relating these different apportionments.
We will specify our apportionment methods by giving two components: first, a probability distribution for selecting some outcome $\omega$ from some suitable set $\Omega$ of possible outcomes; and, second, the apportionment solution $F^{\omega}$ parameterized by $\omega$.
We will refer to such an apportionment method as $F$, leaving the distribution over $\omega$ implicit.
We treat an apportionment method $F$ as a solution-valued random variable, so that $F(\vec{p}, h)$ refers to the method's random apportionment for $\vec{p}, h$ and $F_i(\vec{p}, h)$ refers to the random number of seats apportioned to state $i$.
(We can ignore the measure-theoretic complications of this statement as long as, for each $\vec{p}, h$, the random apportionment $F^{\omega}(\vec{p}, h)$ is a valid random variable, which is the case for all natural constructions following the two-component structure.)
Our axioms, described in the next paragraph, constrain both the random behavior of $F$ and the consistency of any solution $F^{\omega}$ in $F$'s support across inputs.

A method $F$ satisfies \emph{ex ante proportionality} if, for any $\vec{p}$, $h$ and for any state $i$, $i$'s expected number of seats equals $i$'s standard quota, i.e., if $\mathbb{E}[F_i(\vec{p}, h)] = q_i$, where the expected value is over the random choice of apportionment solution.
A method $F$ satisfies \emph{quota}, \emph{house monotonicity}, or \emph{population monotonicity} if all solutions in the method's support satisfy the respective axiom.
In this paper, we mainly search for apportionment methods that combine quota and ex ante proportionality\,---\,the two axioms obtained by \citet{Grim04}\,---\,with either population or house monotonicity.

\section{Population Monotonicity}
\label{sec:popmon}

\subsection{Population Monotonicity Is Incompatible with Quota}
\label{sec:popmonnegative}
We begin by showing that no apportionment method satisfies population monotonicity, quota, and ex ante proportionality.
In fact, quota and population monotonicity alone are incompatible:
We will show that no solution satisfies these two axioms. Since a method satisfying quota and population monotonicity would be a random choice over such solutions, no such method exists either.

At first glance, the incompatibility of quota and population monotonicity might seem to follow from existing results, but these results implicitly make assumptions that are not appropriate for randomized apportionment.
Indeed, \citet{BY82}, who originally proved this incompatibily, as well as variations of their proof~\citep{RU10,El-Helaly19a} all assume what \citet{RU10} call the \emph{order-preserving property}, i.e., if state $i$ has strictly larger population than state $j$, then $i$ must receive at least as many seats as $j$. This property is usually proved as a consequence of neutrality together with population monotonicity.

While the order-preserving property is reasonable for developing deterministic apportionment methods, it is not desirable for the component solutions of a randomized apportionment method.
This is clear for $h=1$: The order-preserving property would mean that only the very largest state(s) can get a seat with positive probability; by contrast, the strength of randomization is that it allows us to allocate the seat to smaller states with some positive probability.
To our knowledge, the existence of quota and population monotone solutions without the assumption of the order-preserving property was an open problem.

\begin{table}[]
    \centering
    \begin{tabular}{rcrrcrrcrr}
    \toprule
    &\hspace{.95cm}& \multicolumn{2}{c}{profile $\vec{p}^A$} &\hspace{.95cm}& \multicolumn{2}{c}{profile $\vec{p}^B$} &\hspace{.95cm}& \multicolumn{2}{c}{profile  $\vec{p}^C$} \\
    \cmidrule{3-4} \cmidrule{6-7} \cmidrule{9-10}
    	state $i$ && $p_i^A$ & $q_i^A$ && $p_i^B$ & $q_i^B$ && $p_i^C$ & $q_i^C$ \\
    	\midrule
        1 && 824 & 8.24 && 824 & 6.99 && 824 & 9.02 \\
        2 && 44 & 0.44 && 44 & 0.37 && 1 & 0.01 \\
        3 && 44 & 0.44 && 44 & 0.37 && 1 & 0.01 \\
        4 && 44 & 0.44 && 44 & 0.37 && 44 & 0.48 \\
        5 && 44 & 0.44 && 222 & 1.88 && 44 & 0.48 \\
    \bottomrule
    \end{tabular}
    \vspace{2mm}
    \caption{Populations and standard quotas for three population profiles, used in showing that population monotonicity and quota are incompatible. The house size is $h=10$.}
    \label{tab:popmon}
\end{table}
\begin{theorem}
\label{prop:quotapopmon}
No (deterministic) apportionment solution satisfies population monotonicity and quota.
\end{theorem}
\begin{proof}
Fix a set of $5$ states, and let $f$ be a solution satisfying quota.
We will show that $f$ must violate population monotonicity by analyzing three types of population profiles, which are given in \cref{tab:popmon}, all for house size $h=10$.
The starting profile is $\vec{p}^A$ in this table.
By quota, state 1 must receive either 8 or 9 seats on this profile, but we will show that either choice leads to a violation of population monotonicity:
First, we show that allocating 9 seats implies a violation of population monotonicity with respect to profile $\vec{p}^B$; second, we show that allocating 8 seats contradicts population monotonicity with respect to $\vec{p}^C$. \medskip

\emph{Allocating 9 seats contradicts population monotonicity.}
Suppose that $f_1(\vec{p}^A, 10) = 9$.
Then, the remaining seat must be given to either state 2, 3, 4, or 5.
Without loss of generality, we may assume that $f(\vec{p}^A, 10) = (9, 0, 0, 0, 1)$.

Next, consider the profile $\vec{p}^B$.
Since quota prevents us from allocating more than 7 seats to state~1 or more than 2 seats to state~5, at least one of the states 2, 3, and 4 must receive a seat on $\vec{p}^B$.
Thus, this state's allocation strictly increases from its allocation of zero seats on $\vec{p}^A$, even though the state's population has not changed.
Moreover, state~1 can receive at most 7 seats on this profile by quota, which is strictly below the 9 seats on $\vec{p}^A$, and state 1's population has also remained the same.
But population monotonicity forbids there to be a pair of states with unchanged population, such that one gains a seat and the other loses a seat.
Hence, if state~1 receives 9 seats on $\vec{p}^A$, then $f$ violates population monotonicity. \medskip

\emph{Allocating 8 seats contradicts population monotonicity.}
Now, suppose that $f_1(\vec{p}^A, 10) = 8$.
The remaining two seats must be given to two states out of 2, 3, 4, and 5; without loss of generality, we may assume that $f(\vec{p}^A, 10) = (8, 0, 0, 1, 1)$.

On profile $\vec{p}^C$, quota implies that state~1 receives at least 9 seats\emdash{}strictly more than the 8 given on $\vec{p}^A$ even though the population has not changed.
Given that there is at most one more seat to hand out, at least one state out of 4 and 5 must receive zero seats on $\vec{p}^C$, which is a strict reduction with respect to $\vec{p}^A$ even though the state's population is the same.
Thus, allocating 8 seats to state~1 on $\vec{p}^A$ also leads to a violation of population monotonicity. \medskip

Since both possible choices for $f_1(\vec{p}^A, 10)$ imply a monotonicity violation, no solution can satisfy both quota and population monotonicity.
\end{proof}

The proof of \cref{prop:quotapopmon} uses $n = 5$ states. and can easily be generalized to any larger number of states. For $n = 3$ states, the Sainte-Laguë method (also known as Webster's method) satisfies both properties \citep[Proposition~6.2]{BY82}. Whether impossibility holds for $n = 4$ remains an open question. The original impossibility result of \citet{BY82} (which, in addition, assumes the order-preserving property) uses only $n = 4$ states.

\subsection{A Population Monotone and Ex Ante Proportional (But Not Quota) Method}
\label{sec:popmonpositive}
The incompatibility between population monotonicity and quota leaves open whether there are apportionment methods satisfying population monotonicity and ex ante proportionality.
The answer is positive, and we will now construct a method satisfying both axioms.

Since population monotonicity relates a solution's result for one input with those for infinitely many other inputs, it highly constrains the shape of population monotone solutions.
In fact, under widely assumed regularity conditions, population monotone solutions are exactly characterized~\citep{BY82} by the class of \emph{divisor methods} (for consistency with our terminology, \emph{divisor solutions}), which will inspire our randomized method.
A divisor solution is defined by a \emph{divisor criterion}, which is a monotone increasing function $d : \natszero \to \mathbb{R}_{\geq 0}$.
(For instance, the Huntington-Hill solution is induced by $d(t) \coloneqq \sqrt{t \, (t + 1)}$.)
For a population profile $\vec{p}$ and house size $h$, the divisor solution corresponding to $d$ can be calculated by considering the sets $\{p_i / d(t) \mid t \in \natszero\}$ for each state $i$, determining the $h$ largest values across all sets, and allocating to each state $i$ a number of seats equal to how many of the $h$ largest values came from $i$'s set.

When state populations change, the values $p_i / d(t)$ evolve in a way that ensures population monotonicity: if state $i$ grows and state $j$ shrinks, a value belonging to state $i$ might overtake a value belonging to $j$ that was originally larger, but none of $j$'s values can overtake any of $i$'s values. As a result, $i$'s number of seats cannot decrease if $j$'s number of seats increases.

To avoid the order-preserving property (for reasons described in \cref{sec:popmonnegative}), we will (randomly) choose a different divisor criterion for each state, which leaves the above argument for population monotonicity intact.
The question is how to sample these divisor criteria such that ex ante proportionality holds.
The answer lies in the properties of Poisson arrival processes: Across independent, scaled Poisson processes, which process yields the first arrival is distributed with probabilities that are proportional to the reciprocals of the scaling factors; and, since the interarrival times of the processes are memoryless, the same holds for each subsequent arrival.
These properties allow us to randomly construct a ``generalized divisor solution'' such that the overall distribution satisfies population monotonicity and ex ante proportionality.

\begin{theorem}
There exists an apportionment method $F$ that satisfies population monotonicity and ex ante proportionality.
\label{thm:proppopmon}
\end{theorem}
\begin{proof}
Which solution is randomly chosen by the method will depend on the values taken on by $n$ independent Poisson arrival processes with rate $1$, which we define as our outcome $\omega$.
We now construct the solution $F^{\omega}$ corresponding to any given $\omega$.
For each state $i$, $\omega$ determines an infinite sequence $0 < x_1^i < x_2^i < \dots$ of arrival times.
We will describe the apportionment given by $F^{\omega}$ on input $\vec{p}$ and $h$, which we illustrate in \cref{fig:poisson}:
First, we divide each arrival time $x_t^i$ by the corresponding state's population, i.e., we set $y_t^i \coloneqq x_t^i / p_i$.
Second, we combine the $y_t^i$ for all $t$ and $i$ in a single arrival sequence $(z_1, i_1), (z_2, i_2), \dots$ labeled with states, i.e., each $(z_j, i_j)$ corresponds to some arrival $y_t^i$ for some $i$ and $t$, such that $z_j = y_t^i$ is the arrival time, $i_j=i$ is the agent label, and the $z_j$ are sorted in increasing order.
Third, we allocate $|\{1 \leq j \leq h \mid i_j = i\}|$ many seats to each state $i$, i.e., a number of seats equal to how many among the $h$ smallest scaled arrival times belonged to $i$'s arrival process.
This specifies the solution $F^{\omega}$, and, hence, the entire apportionment method $F$.
Note that $F^{\omega}$ closely resembles a divisor solution, where state $i$'s set is $\{1 / y_t^i \mid t \in \natszero\} = \{p_i / x_t^i \mid t \in \natszero\}$, i.e., where, for each state $i$, $t \mapsto x_t^i$ plays the role of a state-specific divisor criterion.

\definecolor{myblue}{rgb}{0.0, 0.45, 0.73}
\begin{figure}
    \centering
    \begin{tikzpicture}[yscale=0.6, xscale=1.8,
    axis label/.style={font=\footnotesize, anchor=east},
    point label/.style={font=\footnotesize, anchor=north, %
    }]
    	\node [axis label] at (-0.3,0) {state 1:};
    	\draw [draw=black!50, thick, -stealth] (0,0) -- (5,0);
    	\draw [draw=black, thick, -stealth] (0,-1) -- (5,-1);
    	
    	\node [axis label] at (-0.3,-2.5) {state 2:};
    	\draw [draw=black!50, thick, -stealth] (0,-2.5) -- (5,-2.5);
    	\draw [draw=black, thick, -stealth] (0,-3.5) -- (5,-3.5);
    	
    	\node [axis label] at (-0.3,-5) {combined:};
    	\draw [draw=black, thick, -stealth] (0,-5) -- (5,-5);

    	\foreach \num/\x/\z in {1/1/3,2/3/6,3/4/7} {
    		\node [fill=black!70, inner sep=1.7pt, circle] at (\x, 0) {};
    		\node[point label] at (\x, 0) {$x_\num^1$};
    		
    		\node [fill=orange!90!black, inner sep=1.7pt, circle] at (0.7*\x, -1) {};
    		\node[point label] at (0.7*\x, -1) {$y_\num^1$};
    		
    		 \node [fill=orange!90!black, inner sep=1.7pt, circle] at (0.7*\x, -5) {};
    		 \node[point label] at (0.7*\x, -5) {$z_{\z}\vphantom{z_1^1}$};
    		 
    	};
    	\foreach \num/\x/\z in {1/0.7/1,2/1.5/2,3/3.5/4,4/4.5/5} {
    		\node [fill=black!70, inner sep=2pt] at (\x, -2.5) {};
    		\node[point label] at (\x, -2.5) {$x_\num^2$};
    		
    		\node [fill=myblue, inner sep=2pt] at (0.3*\x, -3.5) {};
    		\node[point label] at (0.3*\x, -3.5) {$y_\num^2$};
    		
    		\node [fill=myblue, inner sep=2pt] at (0.3*\x, -5) {};
    		\node[point label] at (0.3*\x, -5) {$z_{\z}\vphantom{z_1^1}$};
    	};

    	\draw [decorate, thick, decoration = {brace, mirror}] (0.1,-5.9) --  (1.5,-5.9);
    	\node [font=\footnotesize, %
    	] at (0.8,-6.5) {apportionment for $h=5$};
    \end{tikzpicture}
    \caption{Illustration of the population-monotone method in \cref{thm:proppopmon}.}
    \label{fig:poisson}
\end{figure}
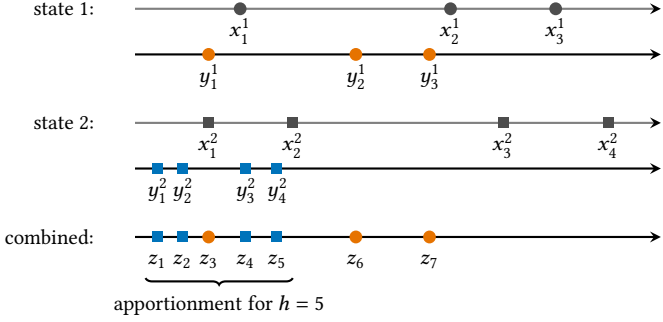

First, we show that $F$ satisfies ex ante proportionality.
For this, fix some $\vec{p}$ and $h$.
Then, the $\{y_t^i\}_{t \geq 1}$ for each $i$ are distributed as the arrival sequences of independent Poisson processes, where $i$'s arrival process has a rate of $p_i$.
As stated by the coloring theorem for Poisson processes~\citep[p.~53]{Kingman93}, our labeled arrival sequence $(z_j, i_j)$ has the same distribution as if we had sampled a Poisson arrival process $0 < z_1 < z_2 < \dots$ with arrival rate $\sum_{i \in N} p_i$ and had drawn each $i_j$ independently, choosing each $i \in N$ with probability proportional to $p_i$.
Since the $z_j$ and $i_j$ are independent in this way, $F(\vec{p}, h)$ is distributed as if sampling $h$ states, with probability proportional to the states' populations and with replacement.
In particular, this implies ex ante proportionality.

It remains to show that $F$ satisfies population monotonicity.
Fix an $\omega$, i.e., the $x_t^i$, as well as two inputs $\vec{p}, h$ and $\pvec{p}', h'$, for which we will show that $F^{\omega}$ does not exhibit a population paradox.
Denoting the inputs' respective variables by $y_t^i, z_j$ and ${y_t^i}', z_j'$, it is easy to see that, for all $i$ for which $p_i' \geq p_i$, ${y_t^i}' \leq y_t^i$ for all $t$, and that, for all $i$ for which $p_i' \leq p_i$, ${y_t^i}' \geq y_t^i$ for all $t$.
Observe that each state $i$ receives a number of seats equal to the number of its scaled arrival times $y_t^i$ (resp., ${y_t^i}'$) that are at most $z_h$ (resp., $z_h'$).

Suppose that $z_h' \geq z_h$ (the reasoning for the case $z_h' \leq z_h$ is symmetric).
Then, whenever $y_t^i \leq z_h$ for a state~$i$ for which $p_i' \geq p_i$, then ${y_t^i}' \leq y_t^i \leq z_h \leq z_h'$, which shows that $i$'s seat number must weakly increase.
One verifies that this rules out a population paradox on $\vec{p},h$ and $\pvec{p}', h'$.
Together with the symmetric argument for $z_h' \leq z_h$, this establishes population monotonicity.
\end{proof}

Clearly, the solutions' resemblance to divisor solutions enabled our proof of population monotonicity.
At the same time, using different ``divisor criteria'' for different states allowed to avoid the order-preserving property, which would have prevented ex ante proportionality as described in \cref{sec:popmonnegative}.
Less satisfying is that, whereas classic divisor criteria satisfy bounds such as $t \leq d(t) \leq t+1$, the ``divisor criteria'' used in the last theorem do not satisfy any such bounds.
As a result, solutions are likely to substantially deviate from proportionality ex post.
An interesting question for future work is whether \cref{thm:proppopmon} can be strengthened to additionally satisfy lower quota or upper quota.

\section{House Monotonicity}
\label{sec:housemon}
While we cannot obtain population monotonicity without giving up on quota, we now propose an apportionment method that combines \emph{house monotonicity} with quota and ex ante proportionality.

\subsection{Examples of Pitfalls}
\label{sec:pitfalls}
An intuitive strategy for constructing a house monotone randomized apportionment method is to do it inductively, seat-by-seat. Thus, we would need a strategy for extending a method that works for all house sizes $h' \le h$ to a method that also works for house size $h+1$. In this section, we give examples suggesting that this does not work, by showing that some reasonable methods cannot be extended without violating quota or ex ante proportionality. This motivates a search for a more ``global'' strategy for constructing a house-monotone method.

\smallskip
\emph{Example 1.} Our first example will show that there are apportionments for a given $h$ that satisfy quota, but that are ``\emph{toxic}'' in that they can never be chosen by a house monotone solution which satisfies quota.
Suppose that we have four states with populations $\vec{p} = (1, 2, 1, 2)$.
The distribution that we will consider is the one given by Grimmett's method \citep{Grim04} (as described in the introduction) for these inputs.
Let $h = 2$. Observe that, if the random permutation chosen by Grimmet's method is identity and if furthermore $U > 2/3$, then Grimmett's method will return the allocation $(1, 0, 1, 0)$.
But we will show that no solution $f$ such that $f(\vec{p}, 2) = (1,0,1,0)$ can satisfy house monotonicity and quota.
Indeed, if $f$ is house monotone, then at least one out of state 2 or state 4 must still be given zero seats by $f$ when $h=3$, but quota requires that both states receive exactly one seat when $h=3$.
It follows that Grimmett's method, the apportionment method of \citet{ALM+19}, or any other method satisfying quota and whose support contains solutions $f$ with $f(\vec{p}, 2) = (1,0,1,0)$, cannot be house monotone.
\smallskip

Thus, a first challenge that any quota and house monotone method must overcome is to never produce a toxic apportionment for a specific $h$ that cannot be extended to all larger house sizes in a house monotone and quota-compliant way.
\citet{Still79} and later \citet{BY79} give a characterization of non-toxic apportionments, but we found no way of transforming this characterization into an apportionment method that would be ex ante proportional.

\smallskip
\emph{Example 2.}
Our second example shows that, even if there are no toxic apportionments in the support of a distribution, the wrong distribution over apportionments might still lead to violations of one of the axioms.
Let there be four states with populations $\vec{p} = (45, 25, 15, 15)$ and let $h=3$; thus, the standard quotas are $(1.35, 0.75, 0.45, 0.45)$. We consider the following distribution over allocations:
\[ \vec{a} = \left\{\begin{array}{@{}ll@{\qquad}ll@{}}
(2, 1, 0, 0) & \text{with probability $35\%$,} &  (1, 1, 0, 1) & \text{with probability $20\%$, and} \\
(1, 1, 1, 0) & \text{with probability $20\%$,} & (1, 0, 1, 1) & \text{with probability $25\%$.}
\end{array}\right.
\]
As we show in \cref{app:pitfalls}, none of these allocations is toxic, and the distribution can be part of an apportionment method in which all three axioms hold for $\vec{p}$ and all $h' \leq 3$.
Nevertheless, we show in the following that any apportionment method $F$ that satisfies house monotonicity and quota and that has the above distribution for $F(\vec{p}, 3)$ must violate ex ante proportionality.
Indeed, fix such an $F$.
On the one hand, note that, for $h=4$, state 2's standard quota is $4 \cdot \frac{25}{100} = 1$, so any quota apportionment must give the state 1 seat.
Since any solution $f$ in the support of $F$ satisfies house monotonicity and quota by assumption, any $f$ such that $f(\vec{p}, 3) = (1, 0, 1, 1)$ must satisfy $f(\vec{p}, 4) = (1, 1, 1, 1)$.
Thus, with at least $25\%$ probability, $F_1(\vec{p}, 4) = 1$.
On the other hand, since state 1's standard quota for $h=4$ is $1.8 \leq 2$, $F_1(\vec{p}, 4) \leq 2$ holds deterministically, by quota.
It follows that $\mathbb{E}[F_1(\vec{p}, 4)] \leq 25\% \cdot 1 + 75\% \cdot 2 = 1.75 < 1.8$, which means that $F$ must violate ex ante proportionality as claimed.
To avoid this kind of conflict between house monotonicity, ex ante proportionality, and quota, the distribution of $F(\vec{p}, 3)$ must allocate at least $5\%$ combined probability to the allocations $(2, 0, 1, 0)$ and $(2, 0, 0, 1)$, which to us is not obvious other than by considering the specific implications on $h=4$ as above.

\begin{comment}
Have 5 states with population sizes p₁ = 45%
According to my program, by iteratively extending the flow, one can arrive at the following distribution over apportionments of house size 3:
(2, 1, 0, 0, 0) with probability 35%
(1, 1, 1, 0, 0) with probability 15%
(1, 1, 0, 1, 0) with probability 15%
(1, 1, 0, 0, 1) with probability 10%
(1, 0, 1, 0, 1) with probability 10%
(1, 0, 0, 1, 1) with probability 10%
(1, 0, 1, 1, 0) with probability 5%

Based on this distribution, there is no way to complete the flow for h=4:
Indeed, we will show that it is not possible to route 45%
\end{comment}

\subsection{Cumulative Rounding}
\label{sec:cumulativeintro}
The examples of the last section showed that it is difficult to construct house monotone apportionment methods seat-by-seat.
In this section, we develop an approach that is able to explicitly take into account how rounding decisions constrain each other across house sizes.
Our approach will be based on dependent randomized rounding in a bipartite graph that we construct.
First, we state the main theorem by \citet{GKP+06}:
\begin{theorem}[\citeauthor{GKP+06}]
\label{thm:dependent}
Let $(A\cup B, E)$ be an undirected bipartite graph with bipartition $(A, B)$.
Each edge $e \in E$ is labeled with a weight $w_e \in [0, 1]$. For each $v \in A \cup B$, we denote the \emph{fractional degree} of $v$ by $d_v \coloneqq \sum_{v \in e \in E} w_e$.

Then there is a random process, running in $\mathcal{O}\big((|A| + |B|) \cdot |E|\big)$ time, that defines random variables $X_{e} \in \{0,1\}$ for all $e \in E$ such that the following properties hold:
\begin{description}
    \item[Marginal distribution:] for all $e \in E$, $\mathbb{E}[X_{e}] = w_{e}$,
    \item[Degree preservation:] for all $v \in A \cup B$, $\sum_{v \in e \in E} X_e \in \{ \lfloor d_v \rfloor, \lceil d_v \rceil \}$, and
    \item[Negative correlation:] for all $v \in A \cup B$ and $S \subseteq \{e \in E \mid v \in e\}$, $\mathbb{P}[\bigwedge_{e \in S} X_e = 1] \leq \prod_{e \in S} w_e$ and $\mathbb{P}[\bigwedge_{e \in S} X_e = 0] \leq \prod_{e \in S} (1 - w_e)$.
\end{description}
\end{theorem}
If $X_e = 1$ for an edge $e$, we say that $e$ gets \emph{rounded up}, and if $X_e = 0$ then $e$ gets \emph{rounded down}.
We do not use negative correlation in our apportionment results, but it is crucial in many applications of dependent rounding since it implies that linear combinations of the shape $\sum_{e \in S} a_e \, X_e$ for some $a_e \in [0,1]$ obey Chernoff concentration bounds~\citep{PS97}.

To see the connection to apportionment, let $\vec{p}$ be a population profile.
Then to warm up, the problem of apportioning a single seat can be easily cast as dependent rounding in a bipartite graph:
Indeed, let $A$ consist of a single special node $a$ and let $B$ contain a node $b_i$ for each state $i$. We draw an edge $e = \{a, b_i\}$ with weight $w_{e} = p_i / \sum_{j \in N} p_j$ for each state $i$.
Apply dependent rounding to this star graph. 
Then $a$'s fractional degree of exactly 1 means that, by degree preservation, exactly one edge $\{a, b_i\}$ gets rounded up, which we interpret as the seat being allocated to state $i$.
Moreover, marginal distribution ensures that each state receives the seat with probability proportional to its population.
This shows that randomized rounding can naturally express ex ante proportionality, which will become a useful building block in the following.

Next, we will expand our construction to multiple house seats, and to satisfying house monotonicity across different house sizes.
The most natural way is to duplicate the star graph from the last paragraph, once per house size $h = 1, 2, \dots$ with nodes $a^h$, $\{b_i^h\}_{i \in N}$ and edges $\big\{\{a^h, b_i^h\}\big\}_{i \in N}$.
(In this intuitive exposition, we will not consider any explicit upper bound on the house sizes. Our formal result in \cref{thm:housemono} will round a finite graph but this will suffice to obtain house monotonicity for all house sizes $h \in \natsone$.)
If $\{a^h, b_i^h\}$ gets rounded up in the $h$-th copy of the star graph, we interpret this as the $h$-th seat going to state $i$.
In other words, we determine how many seats get apportioned to state $i$ for a house size $h$ by counting how many edges $\{a^{h'}, b_i^{h'}\}$ got rounded up across all $h' \leq h$.
This construction automatically satisfies house monotonicity, and satisfies ex ante proportionality by the marginal distribution property, but may violate quota by arbitrary amounts.

To explain how randomized rounding might be useful for guaranteeing quota, let us give a few details on how \citeauthor{GKP+06}'s \emph{pipage rounding} procedure randomly rounds a bipartite graph.
In each step, pipage rounding selects either a cycle or a maximal path consisting of edges with \emph{fractional} weights in $(0, 1)$.
The edges along this cycle or path are then alternatingly labeled ``even'' or ``odd'', which is possible because each cycle in a bipartite graph has an even number of edges.
Depending on a biased coinflip and appropriate numbers $\alpha, \beta > 0$, the algorithm either (1) increases all odd edge weights by $\alpha$ and decreases all even edge weights by $\alpha$, or (2) decreases all odd edge weights by $\beta$ and increases all even edge weights by $\beta$.
In this process, more and more edge weights become zero or one, which determines the $X_e$ once no fractional edges remain.

The cycle/path rounding steps in pipage rounding represent an opportunity to couple the seat-allocation decisions across $h$, in a way that ultimately will allow us to guarantee quota.
In our current graph consisting of disjoint stars, there are no cycles and the maximal paths are always pairs of edges $\{a^h, b_i^h\}, \{a^h, b_j^h\}$ for two states $i, j$ and some $h$.
Thus, pipage rounding correctly anti-correlates the decision of giving the $h$-th seat to state $i$ and the decision of giving the $h$-th seat to state $j$, but decisions for different seats remain independent.
To guarantee quota, increasing (resp., decreasing) the probability of the $h$-th seat going to state $i$ should also decrease (resp., increase) the probability of some nearby seats $h'$ going to state $i$ and increase (resp., decrease) the probability of seats $h'$ going to some other state $j$.
The difficulty is to choose these $h'$ and $j$ to provide quota, which is particular tricky since, in the course of running pipage rounding, some of the edge weights will be rounded to zero and one and no longer be available for paths or cycles.

Not only are we able to use pipage rounding to guarantee quota, but we will do so through a general construction that adds quota-like guarantees to an arbitrary instance of repeated randomized rounding; we refer to this technique as \emph{cumulative rounding}. In the following statement, the ``time steps'' $t$ take the place of our possible house sizes $h$.

\begin{theorem}
\label{thm:cumulative}
Let $(A \cup B, E)$ be an undirected bipartite graph.
For each time step $t=1, \dots, T$, consider a set of edge weights $\{w_e^t\}_{e \in E}$ in $[0, 1]$ for this bipartite graph.
For each $v \in A \cup B$ and $1 \leq t \leq T$, we denote the fractional degree of $v$ at time $t$ by $d_v^t \coloneqq \sum_{v \in e \in E} w_e^t$.

Then there is a random process, running in $\mathcal{O}(T^2 \cdot \big(|A| +|B|) \cdot |E|\big)$ time, that defines random variables $X_{e}^t \in \{0,1\}$ for all $e \in E$ and $1 \leq t \leq T$, such that the following properties hold for all $1 \leq t \leq T$. Let $D_v^t \coloneqq \sum_{v \in e \in E} X_e^t$ denote the random degree of $v$ at time $t$.
\begin{description}
    \item[Marginal distribution:] for all $e \in E$, $\mathbb{E}[X_{e}^t] = w_{e}^t$,
    \item[Degree preservation:] for all $v \in A \cup B$, $D_v^t \in \{ \lfloor d_v^t \rfloor, \lceil d_v^t \rceil \}$,
    \item[Negative correlation:] for all $v \in A \cup B$ and $S \subseteq \{e \in E \mid v \in e\}$, $\mathbb{P}[\bigwedge_{e \in S} X_e^t = 1] \leq \prod_{e \in S} w_e^t$ and $\mathbb{P}[\bigwedge_{e \in S} X_e^t = 0] \leq \prod_{e \in S} (1 - w_e^t)$,
    \item[Cumulative degree preservation:] for all $v \in A \cup B$, $\sum_{t' = 1}^t D_v^{t'} \in \{ \lfloor \sum_{t'=1}^{t} d_v^{t'} \rfloor, \lceil \sum_{t' = 1}^{t} d_v^{t'} \rceil \}$.
\end{description}
\end{theorem}
The first three properties could be achieved by applying \cref{thm:dependent} in each time step independently. Cumulative rounding correlates these  rounding processes such that cumulative degree preservation (a generalization of quota) is additionally satisfied.

\subsection{House Monotone, Quota-Compliant, and Ex Ante Proportional Apportionment}
Before we prove \cref{thm:cumulative}, we will explain how cumulative rounding can be used to construct an apportionment method that is house monotone and satisfies quota and ex ante proportionality.

None of these three axioms connects the outcomes at different population profiles $\vec{p}$ and so it suffices to consider them independently. Thus, let us fix a population profile $\vec{p}$. Denote the total population by $p \coloneqq \sum_{i \in N} p_i$.
The behavior of a house monotone solution on inputs with profile $\vec{p}$ and arbitrary house sizes can be expressed through what we call an \emph{infinite seat sequence}, an infinite sequence $\alpha = \alpha_1, \alpha_2, \dots$ over the states $N$.
We will also define \emph{finite seat sequences}, which are sequences $\alpha = \alpha_1, \dots, \alpha_p$ of length $p$ over the states.
Either sequence represents that, for any house size $h$ (in the case of a finite seat sequence: $h \leq p$), the sequence apportions $a_i(h) \coloneqq |\{1 \leq h' \leq h \mid \alpha_{h'} = i\}|$ seats to each state $i$. %
We can naturally express the quota axiom for seat sequences:
$\alpha$ satisfies \emph{quota} if, for all $h$ ($h \leq p$ if $\alpha$ is finite) and all states $i$, we have $a_i(h) \in \{\lfloor h \, p_i / p \rfloor, \lceil h \, p_i / p \rceil\}$.

The main obstacle in obtaining a house monotone method via cumulative rounding is that we can only apply cumulative rounding to a finite number $T$ of copies, whereas the quota axiom must hold for all house sizes $h \in \natsone$.
However, it turns out that for our purposes of satisfying quota, we can treat the allocation of seats $1, 2, \dots, p$ separately from the allocation of seats $p+1, \dots, 2\,p$, the allocation of seats $2\,p+1, \dots 3\,p$, and so forth.
The reason is that, when $h$ is a multiple $k \, p$ of $p$ (for some $k \in \natsone$), each state $i$'s standard quota is an integer $k \, p_i$. Thus, any solution that satisfies quota is forced to choose exactly the allocation $(k \, p_1, \dots, k \, p_n)$ for house size $h$.
At this point, the constraints for satisfying quota and house monotonicity reset to what they were at $h = 1$. We make this precise in the following lemma, proved in \cref{app:housemonoproofs}:

\begin{restatable}{lemma}{lemrepetition}
\label{lem:repetition2}
An infinite seat sequence $\alpha$ satisfies quota iff it is the concatenation of infinitely many finite seat sequences $\beta^1, \beta^2, \beta^3, \dots$ of length $p$ each satisfying quota, i.e.,
\[\alpha = \beta^1_1, \beta^1_2, \dots, \beta^1_p, \beta^2_1, \beta^2_2, \dots, \beta^2_p, \beta^3_1, \dots. \]
\end{restatable}

This lemma allows us to apply cumulative rounding to only $T=p$ many copies of a star graph.
Then, cumulative rounding produces a random matching that encodes a finite seat sequence satisfying quota, and \cref{lem:repetition2} shows that the infinite repetition of this finite sequence describes an infinite seat sequence satisfying quota.
This implies the existence of an apportionment method satisfying all three axioms we aimed for. The formal proof is in \cref{app:housemonoproofs}.

\begin{restatable}{theorem}{thmhousemono}
\label{thm:housemono}
There exists an apportionment method $F$ that satisfies house monotonicity, quota, and ex ante proportionality.
\end{restatable}

\paragraph{Implications for deterministic methods.}
Our construction also increases our understanding of \emph{deterministic} apportionment solutions satisfying house monotonicity and quota:
Indeed, the possible roundings of the bipartite graph constructed for cumulative rounding (plus some minor modifications, described in the proof) turn out to correspond one-to-one to the finite seat sequences satisfying quota.
Together with \cref{lem:repetition2}, this gives a characterization of all seat sequences that satisfy quota, providing a graph-theoretic alternative to the characterizations by \citet{Still79} and \citet{BY79}.

\begin{restatable}{theorem}{propquotatone}
\label{prop:quotatone}
For any population profile $\vec{p}$, we can construct a bipartite graph whose perfect matchings are in one-to-one correspondence with the finite seat sequences satisfying quota.
\end{restatable}
The proof is deferred to \cref{app:housemonoproofs}.
Combined with \cref{lem:repetition2} this characterizes the set of infinite seat sequences satisfying quota and thus the apportionment solutions satisfying house monotonicity and quota.

A qualitative difference to the previous characterizations \citep{Still79,BY79} of apportionment solutions satisfying house monotonicity and quota is that the matchings allow for a \emph{geometric} description as the corner points of the bipartite graph's matching polytope.
Since a fractional matching assigning each edge $\{a, b_i\}$ a weight of $p_i / p > 0$ lies in the interior of this polytope of perfect fractional matchings, one immediate consequence of this characterization (equivalently, of ex ante proportionality in \cref{thm:housemono}) is that, for each state $i$ and $h \in \natsone$, there is a house monotone and quota-compliant solution that assigns the $h$-th seat to $i$.
To our knowledge, this result is not obvious based on the earlier characterizations.
More generally, the polytope characterization might be useful in answering questions such as ``For a set of pairs $(h_1, i_1), (h_2, i_2), \dots, (h_t, i_t)$, is there a population-monotone and quota-compliant solution that assigns the $h_j$-th seat to state $i_j$ for all $1 \leq j \leq t$?''
To answer this question, one can remove the nodes $a^{h_j}$ and $b_{i_j}^{h_j}$ from the graph (simulating that they got matched) and check whether the remaining graph still permits a perfect matching, for example using Hall's marriage theorem~\citep{Hall35}.
Finally, our formulation allows to optimize linear objectives over the space of solutions satisfying quota and house monotonicity. For example, such an optimization formulation might highlight natural quota-compliant and house monotone solutions other than the one by \citet{BY75}.

\paragraph{Computation.}
Before we prove the cumulative rounding result in \cref{sec:cumulativeproof}, let us quickly discuss computational considerations of our house-monotone apportionment method.
While it is possible to run dependent rounding on the constructed graph (for a given population profile $\vec{p}$), the running time would scale in $\mathcal{O}(p^2 \, n^2)$, and the quadratic dependence on the total population $p$ might be prohibitive.
In practice, we see two ways to avoid this computational cost:

First, one might often not require a solution that is house monotone on all possible house sizes $h \in \natsone$; instead, it might suffice to rule out Alabama paradoxes for house sizes up to an upper bound $h_{\mathit{max}}$.
In this case, it suffices to apply cumulative rounding on $h_{\mathit{max}}$ many copies of the graph, leading to a much more manageable running time of $\mathcal{O}(h_{\mathit{max}}^2 \, n^2)$.

A second option would be to apply cumulative rounding on all $p$ copies of the graph, but to stop pipage rounding once all edge weights in the first $h$ copies of the graph are integral, even if edge weights for higher house sizes are still fractional.
This would allow to return an apportionment for inputs $\vec{p}, h$ while randomly determining not a single house-monotone solution, but a conditioned distribution $F^c$ over house-monotone solutions, all of which agree on the apportionment for $\vec{p}$ and $h$.
Since all solutions are house monotone, the expected number of seats for a party will always monotonically increase in $h$ across the conditioned distribution.
Should it become necessary to determine apportionments for larger house sizes, one can simply continue the cumulative-rounding process where it left off.
Since the pipage rounding used to prove \cref{thm:cumulative} leaves open which cycles or maximal paths get rounded next, it seems likely that one can deliberately choose cycles/paths such that the apportionment for the first $h$ seats is determined in few rounds. %

\section{Proof of Cumulative Rounding}
\label{sec:cumulativeproof}
We will now prove \cref{thm:cumulative} on cumulative rounding.
Our proof constructs a weighted bipartite graph including $T$ many copies of $(A \cup B, E)$, connected by appropriate additional edges and nodes, and then applying dependent rounding to this constructed graph.
The additional edges and vertices ensure that if too many edges adjacent to some node $v$ are rounded up in one copy of the graph, then this is counterbalanced by rounding down edges adjacent to $v$ in another copy.
\begin{construction}
\label{constr}
Let $(A \cup B, E)$, $T$, and $\{w_e^t\}_{e,t}$ be given as in \cref{thm:cumulative}.
We construct a new weighted, undirected, and bipartite graph as follows:
For each node $v \in A \cup B$ and for each $t = 1, \dots, T$, create four nodes $v^t$, $\onebar{v}^t$, $\twobar{v}^t$, and $v^{t:t+1}$; furthermore, create a node $v^{0:1}$ for each node $v$.
For each $\{a, b\} \in E$ and $t=1, \dots, T$, connect the nodes $a^t$ and $b^t$ with an edge of weight $w_{\{a, b\}}^t$.
Additionally, for each node $v \in A \cup B$ and each $t = 1, \dots, T$, insert edges with the following weights:
\begin{center}
\begin{tikzpicture}[node distance=.5cm and 4.0cm]
 \tikzstyle{vertex}=[]%
\node[vertex](top) {$v^t$};
\node[vertex,below=of top] (middle) {$\twobar{v}\vphantom{v}^t$};
\node[vertex,below=of middle] (bottom) {$\onebar{v}\vphantom{v}^t$};
\node[vertex,left=of bottom] (left) {$v^{t-1:t}$};
\node[vertex,right=of bottom] (right) {$v^{t:t+1}$};
\draw (left) -- node[below] {\footnotesize $\sum_{t'=1}^{t-1} d_v^{t'} - \left\lfloor \sum_{t'=1}^{t-1} d_v^{t'} \right\rfloor$} (bottom) -- node[below] {\footnotesize $1 - \sum_{t'=1}^{t} d_v^{t'} + \left\lfloor \sum_{t'=1}^{t} d_v^{t'} \right\rfloor$} (right);
\draw (top) -- node[right] {\footnotesize $1 - d_v^t + \left\lfloor d_v^t \right\rfloor$} (middle) -- node[right] {\footnotesize $d_v^t - \left\lfloor d_v^t \right\rfloor$} (bottom);
\end{tikzpicture}
\end{center}
\end{construction}

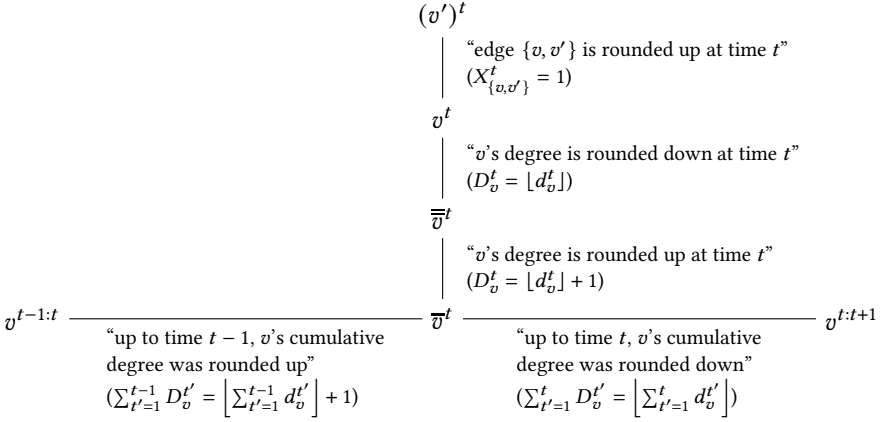
\begin{figure}
    \centering
\begin{tikzpicture}[node distance=.8cm and 4.7cm]
 \tikzstyle{vertex}=[]
\node[vertex](top) {$v^t$};
\node[vertex,above=of top] (toperer) {$(v')^t$};
\node[vertex,below=of top] (middle) {$\twobar{v}\vphantom{v}^t$};
\node[vertex,below=of middle] (bottom) {$\onebar{v}\vphantom{v}^t$};
\node[vertex,left=of bottom] (left) {$v^{t-1:t}$};
\node[vertex,right=of bottom] (right) {$v^{t:t+1}$};
\draw (left) -- node[below,align=left,font=\footnotesize] {``up to time $t-1$, $v$'s cumulative \\ degree was rounded up''\\($\textstyle{\sum_{t'=1}^{t-1} D_v^{t'} = \left\lfloor \sum_{t'=1}^{t-1} d_v^{t'} \right\rfloor + 1}$)} (bottom) -- node[below,align=left,font=\footnotesize] {``up to time $t$, $v$'s cumulative \\ degree was rounded down''\\($\textstyle{\sum_{t'=1}^{t} D_v^{t'} = \left\lfloor \sum_{t'=1}^{t} d_v^{t'} \right\rfloor}$)} (right);
\draw (toperer) -- node[right,align=left,xshift=2mm,font=\footnotesize] {``edge $\{v,v'\}$ is rounded up at time $t$''\\($X_{\{v, v'\}}^t = 1$)} (top) -- node[right,align=left,xshift=2mm,font=\footnotesize] {``$v$'s degree is rounded down at time $t$''\\($D_v^t = \lfloor d_v^t \rfloor$)} (middle) -- node[right,align=left,xshift=2mm,font=\footnotesize] {``$v$'s degree is rounded up at time $t$''\\($D_v^t = \lfloor d_v^t \rfloor + 1$)} (bottom);
\end{tikzpicture}
    \caption{Interpretation of each edge being rounded up in the constructed graph, for arbitrary nodes $v, v' \in A \cup B$ and $1 \leq t \leq T$. The correctness of this characterization will be shown along the proof of \cref{thm:cumulative}, specifically in the sections on degree preservation and cumulative degree preservation.}
    \label{fig:roundinginterpretation}
\end{figure}
Before we go into the proof, we give in \cref{fig:roundinginterpretation} an interpretation for what it means for each edge in the constructed graph to be rounded up.
One can easily verify that, under the (premature) assumption that cumulative rounding satisfies marginal distribution, degree preservation, and cumulative degree preservation, the edge weights coincide with the probabilities of each interpretation's event.
We want to stress that it is not obvious that these descriptions will indeed be consistent for any dependent rounding of the constructed graph, and we will not make use of these descriptions in the proof of \cref{thm:cumulative}.
Instead, the characterizations will follow from intermediate results in the proof.
We give these interpretations here to make the construction seem less mysterious.
We begin the formal analysis of the construction with a sequence of simple observations about the constructed graph (proofs are in \cref{app:cumulativeproofs}).

\begin{restatable}{lemma}{lembipartite}
\label{lem:bipartite}
The graph of \cref{constr} is bipartite.
\end{restatable}

\begin{restatable}{lemma}{lemweightunit}
\label{lem:weightunit}
All edge weights lie between $0$ and $1$.
\end{restatable}

\begin{restatable}{lemma}{lemfractionaldegrees}
\label{lem:fractionaldegrees}
For each node $v \in A \cup B$, the following table gives the fractional degrees of various nodes in the constructed graph, all of which are integer:
\begin{center}
\normalfont
\begin{tabular}{ll}
\toprule
nodes & fractional degree \\
\midrule
$v^t$ $(\forall 1 \! \leq \! t \! \leq \! T)$ & $\lfloor d_v^t \rfloor + 1$ \\
$\onebar{v}^t$ $(\forall 1 \! \leq \! t \! \leq \! T)$ & $\left\lfloor \sum_{t'=1}^t d_v^{t'} \right\rfloor - \left\lfloor \sum_{t'=1}^{t-1} d_v^{t'} \right\rfloor - \lfloor d_v^t \rfloor + 1$ \\ %
$\twobar{v}^t$ $(\forall 1 \! \leq \! t \! \leq \! T)$ & 1 \\
$v^{t:t+1}$ $(\forall 1 \! \leq \! t \! \leq \! T\!-\!1)$ & 1 \\
$v^{0:1}$ & 0 \\
\bottomrule
\end{tabular}
\end{center}
\end{restatable}

\begin{proof}[Proof of \cref{thm:cumulative}]
We define cumulative rounding as the random process that follows \cref{constr} and then applies dependent rounding (\cref{thm:dependent}) to the constructed graph, which is valid since the graph is bipartite and all edge weights lie in $[0,1]$ (\cref{lem:bipartite,lem:weightunit}).
For an edge $e$ in the constructed graph, let $\hat{X}_e$ be the random variable indicating whether dependent rounding rounds it up or down.
For any edge $\{a, b\} \in E$ in the underlying graph and some $1 \leq t \leq T$, we define the random variable $X_{\{a,b\}}^t$ to be equal to $\hat{X}_{\{a^t, b^t\}}$.
Recall that we defined $D_v^t = \sum_{v \in e \in E} X_e^t$.

To prove the theorem, we have to bound the running time of this process, and provide the four guaranteed properties: marginal distribution, degree preservation, negative correlation, and cumulative degree preservation. The last property takes by far the most work.

\paragraph{Running time.}
Without loss of generality, we may assume that each vertex $v \in A \cup B$ is incident to at least one edge, since, otherwise, we could remove this vertex in a preprocessing step.
From this, it follows that $|E| \in \Omega(|A| + |B|)$.
Constructing the graph takes $\mathcal{O}(T \, |E|)$ time, which will be dominated by the time required for running dependent rounding on the constructed graph.
The constructed graph has $(1 + 4 \, T) \, (|A| + |B|) \in \mathcal{O}(T \, (|A| + |B|))$ nodes and $T \, |E| + 4 \, T \, (|A| + |B|) \in \mathcal{O}(T \, |E|)$ edges.
Since the running time of dependent rounding scales in the product of the number of vertices and the number of edges, our procedure runs in $\mathcal{O}(T^2 \, (|A| + |B|) \, |E|)$ time, as claimed.

\paragraph{Marginal distribution.}
For an edge $\{a, b\} \in E$ and $1 \leq t \leq T$, $\mathbb{E}[X_{\{a,b\}}^t] = \mathbb{E}[\hat{X}_{\{a^t,b^t\}}] = w_{\{a, b\}}^t$, where the last equality follows from the marginal-distribution property of dependent rounding.

\paragraph{Degree preservation.}
Fix a node $v \in A \cup B$ and $1 \leq t \leq T$.
By \cref{lem:fractionaldegrees}, the fractional degree of $v^t$ is $\lfloor d_v^t \rfloor + 1$, and thus, by degree preservation of dependent rounding, exactly $\lfloor d_v^t \rfloor + 1$ edges adjacent to $v^t$ must be rounded up.
The only of these edges that does not count into $D_v^t$ is $\{\twobar{v}^t, v^t\}$; depending on whether this edge is rounded up or down, $D_v^t$ is either $\lfloor d_v^t \rfloor$ or $\lfloor d_v^t \rfloor + 1$.
If $d_v^t$ is not integer, the latter number equals $\lceil d_v^t \rceil$, which proves degree preservation.
Else, if $d_v^t$ is an integer, the edge weight of $\{\twobar{v}^t, v^t\}$ is 1.
Dependent rounding always rounds up edges with weight 1, which means that $D_v^t$ is definitely $\lfloor d_v^t \rfloor$ in this case.
Thus, degree preservation holds in either case.

\paragraph{Negative correlation.}
Negative correlation for $v \in A \cup B$, $S \subseteq \{e \in E \mid v \in e\}$, and $1 \leq t \leq T$ directly follows from the negative-correlation property of dependent rounding for the node $v^t$ and the edge set $S' \coloneqq \big\{\{v^t, (v')^t\} \,\big|\, \{v, v'\} \in S\big\}$.

\paragraph{Cumulative degree preservation.}
Fix a node $v \in A \cup B$ and $1 \leq t \leq T$.
We will consider the ``rounded version'' of the constructed graph, i.e., the unweighted bipartite graph over the nodes of the constructed graph in which exactly those edges are present that got rounded up by the randomized rounding process.
We define five sets of nodes in the rounded graph (\cref{fig:cumulativeproof}):
\begin{equation*}
\begin{gathered}
    V \coloneqq \{v^{t'} \mid 1 \leq t' \leq t\} \qquad\qquad
    V' \coloneqq \{(v')^{t'} \mid v' \in (A \cup B) \setminus \{v\}, 1 \leq t' \leq t\} \\
    \onebar{V} \coloneqq \{\onebar{v}^{t'} \mid 1 \leq t' \leq t\} \qquad\qquad
    \twobar{V} \coloneqq \{\twobar{v}^{t'} \mid 1 \leq t' \leq t\} \qquad\qquad
    V^{\bm{:}} \coloneqq \{v^{t:t+1} \mid 0 \leq t' \leq t\}
\end{gathered}
\end{equation*}
For any set of nodes $V_1$ in the rounded graph, we denote its neighborhood by $N(V_1)$, and we will write $\degr(V_1)$ for the sum of degrees of $V_1$ in the rounded graph.
For any two sets of nodes $V_1, V_2$, we write $\cut(V_1, V_2)$ to denote the number of edges between $V_1$ and $V_2$ in the rounded graph.

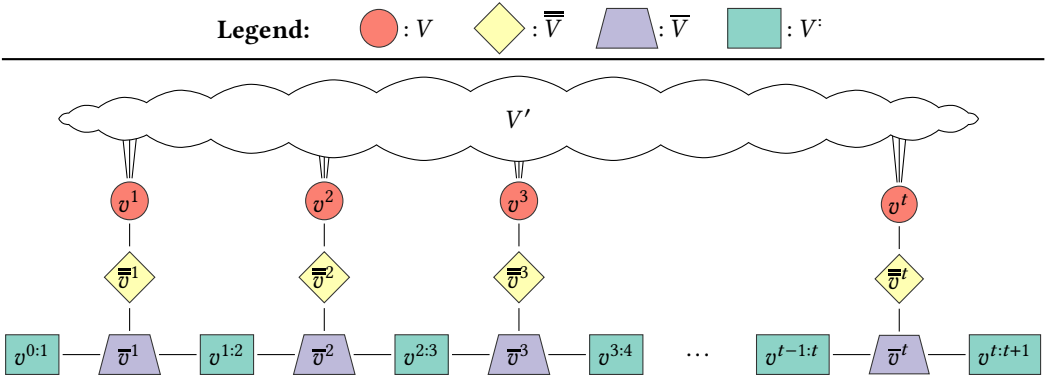
\begin{figure}
    \centering
\begin{tikzpicture}[node distance=.3cm and .5cm]
 \tikzstyle{vertex}=[outer sep=.5mm,draw=black!80]
 \tikzstyle{trans}=[fill=brewer1]
 \tikzstyle{ob}=[trapezium,trapezium angle=75,fill=brewer3]
 \tikzstyle{tb}=[diamond,inner sep=.3mm,fill=brewer2]
 \tikzstyle{norma}=[circle,inner sep=.3mm,fill=brewer4]
 \tikzstyle{legend}=[]

\node[vertex,trans] (v01) {$v^{0:1}$};
\node[vertex,ob,right=of v01] (vb1) {$\onebar{v}^1$};
\node[vertex,trans,right=of vb1] (v12) {$v^{1:2}$};
\node[vertex,ob,right=of v12] (vb2) {$\onebar{v}^2$};
\node[vertex,trans,right=of vb2] (v23) {$v^{2:3}$};
\node[vertex,ob,right=of v23] (vb3) {$\onebar{v}^3$};
\node[vertex,trans,right=of vb3] (v34) {$v^{3:4}$};
\node[right=of v34,xshift=-1mm] (ell) {\dots};
\node[vertex,trans,right=of ell,xshift=-1mm] (vt1t) {$v^{t-1:t}$};
\node[vertex,ob,right=of vt1t] (vbt) {$\onebar{v}^t$};
\node[vertex,trans,right=of vbt] (vtt1) {$v^{t:t+1}$};
\draw (v01)-- (vb1) -- (v12) -- (vb2) -- (v23) -- (vb3) -- (v34);
\draw (vt1t) -- (vbt) -- (vtt1);

\node[vertex,tb,above=of vb1] (vbb1) {$\twobar{v}^1$};
\node[vertex,norma,above=of vbb1] (v1) {$v^1$};
\draw (vb1) -- (vbb1) -- (v1);

\node[vertex,tb,above=of vb2] (vbb2) {$\twobar{v}^2$};
\node[vertex,norma,above=of vbb2] (v2) {$v^2$};
\draw (vb2) -- (vbb2) -- (v2);

\node[vertex,tb,above=of vb3] (vbb3) {$\twobar{v}^3$};
\node[vertex,norma,above=of vbb3] (v3) {$v^3$};
\draw (vb3) -- (vbb3) -- (v3);

\node[vertex,tb,above=of vbt] (vbbt) {$\twobar{v}^t$};
\node[vertex,norma,above=of vbbt] (vt) {$v^t$};
\draw (vbt) -- (vbbt) -- (vt);

\node[yshift=3.1cm] (cl) at ($(v01)!0.5!(vtt1)$) {};

\draw (v1 |- cl) node (tmp) {} -- (v1);
\draw (v1) [xshift=-9mm] -- (tmp);
\draw (v1) [xshift=9mm] -- (tmp);

\draw (v2 |- cl) node (tmp) {} -- (v2);
\draw (v2) [xshift=-9mm] -- (tmp);
\draw (v2) [xshift=9mm] -- (tmp);

\draw (v3 |- cl) node (tmp) {} -- (v3);
\draw (v3) [xshift=-9mm] -- (tmp);
\draw (v3) [xshift=9mm] -- (tmp);

\draw (vt |- cl) node (tmp) {} -- (vt);
\draw (vt) [xshift=-9mm] -- (tmp);
\draw (vt) [xshift=9mm] -- (tmp);

\draw node[cloud, draw,cloud puffs=30,cloud puff arc=70, aspect=17, inner ysep=1.3mm,fill=white] at (cl) (cloud) {$V'$};
\node[yshift=1.15cm,xshift=1.78cm] (leg) at (v1 |- cl) {\textbf{Legend:}};

\node[legend,vertex,norma,right=of leg,yshift=.4mm] (normaleg) {$\phantom{v^1}$};
\node[anchor=west,xshift=-1mm] at (normaleg.east) {: $V$};
\node[legend,vertex,tb,right=of normaleg,xshift=4mm] (tbleg) {$\phantom{\twobar{v}^1}$};
\node[anchor=west,xshift=-1mm] at (tbleg.east) {: $\vphantom{V}\smash{\twobar{V}}$};
\node[legend,vertex,ob,right=of tbleg,xshift=4mm] (obleg) {$\phantom{\onebar{v}^1}$};
\node[anchor=west,xshift=-1mm] at (obleg.east) {: $\vphantom{V}\smash{\onebar{V}}$};
\node[legend,vertex,trans,right=of obleg,xshift=4mm] (transleg) {$\phantom{v^{1:2}}$};
\node[anchor=west,xshift=-1mm] at (transleg.east) (translegtext) {: $\vphantom{V}\smash{V^{\bm{:}}}$};
\draw (v01.west |- tbleg.south) [line width=.8pt]-- (vtt1.east |- tbleg.south);

\end{tikzpicture}
    \caption{Illustration of the counting argument for proving cumulative degree preservation. Edges in the figure are edges from the constructed graph, a superset of the edges in the rounded graph. Node color and shape indicate the set that a node belongs to, as indicated in the legend.}
    \label{fig:cumulativeproof}
\end{figure}
Note that $\sum_{t'=1}^t D_v^{t'}$, which we must bound, equals $\cut(V, V')$.
We will bound this quantity by repeatedly using the following fact, which we refer to \emph{pivoting}:
For pairwise disjoint sets of nodes $V_0, V_1, V_2$, if $N(V_0) \subseteq V_1 \cup V_2$, then $\degr(V_0) = \cut(V_0, V_1) + \cut(V_0, V_2)$.
Since \cref{lem:fractionaldegrees} gives us a clear view of the fractional degrees of nodes in the constructed graph, and since, by degree preservation, a node's degree in the rounded graph must equal the fractional degree whenever the latter is an integer, this property allows us to express cuts in terms of other cuts.
\Cref{fig:cumulativeproof} illustrates which of these sets border on each other, and helps in following along with the derivation.
\begin{align*}
    \textstyle{\sum_{t'=1}^{t} D_v^{t'}} &= \cut(V, V') \\
    &= \degr(V) - \cut(V, \twobar{V}) \tag*{(pivot $V_0 = V, V_1 = V', V_2 = \twobar{V}$)} \\
    &= t + \textstyle{\sum_{t'=1}^t \lfloor d_v^{t'} \rfloor} - \cut(V, \twobar{V}) \tag*{($\degr(V) = t + \textstyle{\sum_{t'=1}^t \lfloor d_v^{t'} \rfloor}$ by \cref{lem:fractionaldegrees})} \\
    &= t + \textstyle{\sum_{t'=1}^t \lfloor d_v^{t'} \rfloor} - \degr(\twobar{V}) + \cut(\twobar{V}, \onebar{V}) \tag*{(pivot $V_0 = \twobar{V}, V_1 = V, V_2 = \onebar{V}$)} \\
    &= \textstyle{\sum_{t'=1}^t \lfloor d_v^{t'} \rfloor} + \cut(\twobar{V}, \onebar{V}) \tag*{($\degr(\twobar{V}) = t$ by \cref{lem:fractionaldegrees})} \\
    &= \textstyle{\sum_{t'=1}^t \lfloor d_v^{t'} \rfloor} +
    \degr(\onebar{V}) - \cut(\onebar{V}, V^{\bm{:}}) \tag*{(pivot $V_0 = \onebar{V}, V_1 = \twobar{V}, V_2 = V^{\bm{:}}$)} \\
    &= \textstyle{\sum_{t'=1}^t \lfloor d_v^{t'} \rfloor} -
    \cut(\onebar{V}, V^{\bm{:}}) + \displaystyle{\sum\nolimits_{t'=1}^t}\textstyle{ \left(\lfloor \sum_{t''=1}^{t'} d_v^{t''}\rfloor - \lfloor\sum_{t''=1}^{t' - 1} d_v^{t''} \rfloor - \lfloor d_v^{t'} \rfloor + 1\right)}  \tag*{(\cref{lem:fractionaldegrees})} \\
    &= \textstyle{\sum_{t'=1}^t \lfloor d_v^{t'} \rfloor} - \cut(\onebar{V}, V^{\bm{:}}) + \lfloor \sum_{t''=1}^t d_v^{t''}\rfloor - \sum_{t'=1}^t \lfloor d_v^{t'} \rfloor + t  \tag*{(telescoping sum)} \\
    &= \textstyle{\lfloor \sum_{t'=1}^t d_v^{t'}\rfloor} + t - \cut(\onebar{V}, V^{\bm{:}}).
\intertext{To bound $\cut(\onebar{V}, V^{\bm{:}})$ in the last expression, observe that $N(V^{\bm{:}} \setminus \{v^{t:t+1}\}) \subseteq \onebar{V}$, from which it follows that $\cut(\onebar{V}, V^{\bm{:}} \setminus \{v^{t:t+1}\})=\degr(V^{\bm{:}} \setminus \{v^{t:t+1}\}) = t-1$.
Thus, $\cut(\onebar{V}, V^{\bm{:}}) = t-1 + \bone\{\hat{X}_{\{\onebar{v}^t, v^{t:t+1}\}}\}$, and we resume the above equality}
 \textstyle{\sum_{t'=1}^{t} D_v^{t'}} &= \textstyle{\lfloor \sum_{t'=1}^t d_v^{t'}\rfloor} + t - (t - 1 + \bone\{\hat{X}_{\{\onebar{v}^t, v^{t:t+1}\}}\}) = \textstyle{\lfloor \sum_{t'=1}^t d_v^{t'}\rfloor} + 1 - \bone\{\hat{X}_{\{\onebar{v}^t, v^{t:t+1}\}}\}.
\end{align*}
If $\sum_{t'=1}^t d_v^{t'}$ is not an integer, the above shows that $\sum_{t'=1}^{t} D_v^{t'}$ is either the floor or ceiling of $\sum_{t'=1}^t d_v^{t'}$, establishing cumulative degree preservation.
Else, if $\sum_{t'=1}^t d_v^{t'}$ is integer, note that the weight of the edge $\{\onebar{v}^t, v^{t:t+1}\}$ in the constructed graph is 1.
Since dependent rounding always rounds such edges up, $\textstyle{\sum_{t'=1}^{t} D_v^{t'}} = \lfloor \sum_{t'=1}^t d_v^{t'}\rfloor$.
This establishes cumulative degree preservation, the last of the properties guaranteed by the theorem.
\end{proof}

\section{Other Applications of Cumulative Rounding}
\label{sec:applications}

Our exploration of house monotone randomized apportionment led us to the more general technique of cumulative rounding, which we believe to be of broader interest. We next illustrate this by discussing additional applications. 

\subsection{Sortition of the European Commission}
\label{sec:europeancommission}
The European Commission is one of the main institutions of the European Union, in which it plays a role comparable to that of a government.
The commission consists of one commissioner from each of the 27 member states, and each commissioner is charged with a specific area of responsibility.
Since the number of EU member states has nearly doubled in the past 20 years, so has the size of the commission.
Besides making coordination inside the commission less efficient, the enlargement of the commission has led to the creation of areas of responsibility much less important than others.
Since the important portfolios are typically reserved for the largest member states, smaller states have found themselves with limited influence on central topics being decided in the commission.

To remedy this imbalance, \citet{BH09} propose to reduce the number of commissioners to 15, meaning that only a subset of the 27 member states would send a commissioner at any given time.
Which states would receive a seat would be determined every 5 years by a weighted lottery (``sortition''), in which states would be chosen with degressive proportional weights. Degressive means that smaller states get non-proportionately high weight; such weights are already used for apportioning the European parliament.
The authors argue that by the law of large numbers, political representation on the commission would be essentially proportional to these weights in politically relevant time spans.

However, a follow-up simulation study by \citet{BHJ13} challenges this assertion on two fronts:
(1)~First, the authors find that their implementation of a weighted lottery chooses states with probabilities that deviate from proportionality to the weights in a way that is not analytically tractable~(see \citealt[p.\ 24]{BH83}).
(2)~Second, and more gravely, their simulations undermine ``a central argument in favor of legitimacy'' in the original proposal, namely, that ``in the long term, the seats on the commission would be distributed approximately like the share of lots''~\citep[own translation]{BHJ13}.
From a mathematical point of view, the authors had overestimated the rate of concentration across the independent lotteries.
Instead, in the simulation,
it takes 30 lotteries (150 years) until there is a probability of 99\% that
all member states have sent at least one commissioner. %

These serious concerns could be resolved by using cumulative rounding to implement the weighted lotteries.
Specifically, we would again construct a star graph with a special node $a$ and one node $b_i$ for each state $i$, setting $T$ to the desired number of consecutive lotteries.
For each $1 \leq t \leq T$, each edge $\{a, b_i\}$ would be weighted by $\frac{15 \, w_i}{\sum_{j \in N} w_j}$ where $w_j$ is state $j$'s degressive weight. %
Degree preservation on $a$ would ensure that in each lottery $t$, exactly 15 distinct states are selected.
By marginal distribution, the selection probabilities would be exactly proportional to the degressive weights, resolving issue~(1).
Furthermore, cumulative degree preservation on the state nodes would eliminate issue~(2).
If we take the effective selection probabilities of \citet{BHJ13} as the states' weights, even the smallest states $i$ would have an edge weight $w_{\{a,b_i\}}^t \approx 0.187$.
Then, cumulative quota prevents any state from getting rounded down in $11 = \lceil 2 / 0.187 \rceil$ consecutive lotteries:
Indeed, fixing any $0 \leq t_0 \leq T - 11$,
\[ \textstyle{\sum_{t'=1}^{t_0 + 11} D_{b_i}^{t'} \geq \lfloor (t_0 + 11) \, 0.187 \rfloor \geq \lfloor t_0 \, 0.187 \rfloor + 2 \geq \lceil t_0 \, 0.187 \rceil + 1 \geq \sum_{t'=1}^{t_0} D_{b_i}^{t'}} + 1, \]
which means that state $i$ must have been selected at least once between time $t_0 + 1$ and $t_0 + 11$.
In political terms, this means that 55, not 150, years would be enough to \emph{deterministically} ensure that each member state send a commissioner at least once in this period.

This cumulative rounding approach can accommodate weights that change across lotteries according to population projections (which \citeauthor{BHJ13} do for some of their experiments), simply by choosing different weights across the copies of the star graph.
It is necessary, however, that these population changes are known in advance, since \emph{no} algorithm can guarantee cumulative degree preservation in an iterated apportionment setting in which population changes are observed online (i.e., just in time for the next apportionment to be made).
This is shown by the following example:
Let there be 4 states, and allocate a single seat per time step. At time $t=1$, all four states have an equal population, thus an edge weight of $1/4$. Without loss of generality, the first seat goes to state 4. At time $t=2$, state 4 disappears while states 1 through 3 have equal population, thus each an edge weight of $1/3$.
Without loss of generality, the second seat goes to state 3.
At time $t=3$, states 3 and 4 have zero population whereas states 1 and 2 have equal population (i.e., an edge weight of $1/2$).
Since states 1 and 2 have a cumulative quota of $1/4+1/3+1/2>1$, cumulative degree preservation requires both states to have at least one among the first three seats, but this is clearly impossible. 
We have presented this argument for an adaptive adversary; for a non-adaptive adversary, essentially the same argument shows that any online apportionment mechanism will violate cumulative degree preservation with probability at least $1/12$.
Note that this impossibility holds even in the absence of marginal distribution and negative correlation.
It is an intriguing question how one should design an online apportionment method that keeps violations from cumulative degree preservation at a minimum, and ensures that the cost or benefit of such deviations is fairly spread across the remaining states.

\subsection{Repeated Allocation of Courses to Faculty Or Shifts to Workers}
\label{sec:faculty}
A common paradigm in fair division is to first create a fractional assignment between agents and resources, and to then implement this fractional assignment in expectation, through randomized rounding.
Below, we describe a setting of allocating courses to faculty members in a university department, in which implementing a fractional assignment using cumulative rounding is attractive.

For a university department, denote its set of faculty members by $A$ and the set of possible courses to be taught by $B$.
For each faculty member $a$ and course $b$, let there be a weight $w_{\{a,b\}} \in [0, 1]$ indicating how frequently course $b$ should be taught by $a$ on average. These numbers could be derived using a process such as probabilistic serial~\citep{BM01}, the Hylland-Zeckhauser mechanism~\citep{HZ79}, or the mechanisms by \citet{BCKM13}, which would transform preferences of the faculty over which courses to teach into such weights. (Although these mechanisms are formulated for goods, they can be applied to bads when the number of bads allocated to each agent is fixed, as it is when allocating courses to faculty or shifts to workers.)
We will allow arbitrary fractional degrees on the faculty side (so one person can teach multiple courses) while assuming that the fractional degree of any course $b$ is at most $1$.

Applying cumulative rounding to this graph (using the same edge weights in each period) for consecutive semesters $1 \leq t \leq T$, we get the following properties.
\begin{itemize}
    \item Marginal distribution implies that, in each semester, faculty member $a$ has a probability $w_{\{a, b\}}$ of teaching course $b$.
    \item Degree preservation on the course side means that a course is never taught by two different faculty members in the same semester.
    \item Degree preservation on the faculty side implies that a faculty member $a$'s teaching load does not vary by more than 1 between semesters; it is either the floor or the ceiling of $a$'s expected teaching load.
    \item Cumulative degree preservation on the course side ensures that courses are offered with some regularity. For example, if a course's fractional degree is $1/2$, it will be taught exactly once in every academic year (either in Fall or in Spring).
    \item Cumulative degree preservation on the faculty side allows for non-integer teaching load. For example, a faculty member with fractional degree $1.5$ will have a ``2-1'' teaching load, i.e., they will teach 3 courses per year, either 2 in the Fall and 1 in the Spring or vice versa.
\end{itemize}

The same approach is applicable for matching workers to shifts.

One could also use cumulative rounding to repeatedly round a fractional assignment of general chores, such as the ones computed by the online platform spliddit.org~\citep{GP14}.
In this case, a caveat is that (cumulative) degree preservation only ensures that the \emph{number} of assigned chores is close to its expected number per time period, not necessarily the \emph{cost} of the assigned chores.
However, if many chores are allocated per time step, and if costs are additive, then an agent's per-timestep cost is well-concentrated, which follows from the negative-correlation property that permits the application of Chernoff concentration bounds~\citep{PS97}.

\section{Discussion}
\label{sec:discussion}
Though our work is motivated by the application of apportioning seats at random, the technical questions we posed and addressed are fundamental to the theoretical study of apportionment.
In a sense, any deterministic apportionment solution is ``unproportional''\emdash{}after all, its role is to decide which agents receive more or fewer seats than their standard quota.
By searching for randomized methods satisfying ex ante proportionality, we ask whether these unproportional solutions can be combined (through random choice) such that these deviations from proportionality cancel out to achieve perfect proportionality, and whether this remains possible when we restrict the solutions to those satisfying subsets of the axioms population monotonicity, house monotonicity, and quota.
Naturally, this objective pushes us to better understand the whole space of solutions satisfying these subsets of axioms, including the space's more extreme elements.
Therefore, it is in hindsight not surprising that our work led to new insights for deterministic apportionment: a more robust impossibility between population monotonicity and quota (\cref{prop:quotapopmon}), an exploration of solutions generalizing the divisor solutions (\cref{thm:proppopmon}), and a geometric characterization of house monotone and quota compliant solutions (\cref{prop:quotatone}).

Concerning the cumulative rounding technique introduced in this paper, we have only scratched the surface in exploring its applications.
In particular, we hope to investigate whether cumulative rounding can extend existing algorithmic results that use dependent rounding, and whether it can be used to construct new approximation algorithms.
For both of these purposes, the negative-correlation property, which we have not used much so far, will hopefully turn out to be valuable.

Despite their advantageous properties, randomized mechanisms have in the past often met with resistance by practitioners and the public~\citep{KPS18}, but we see signs of a shift in attitudes.
\emph{Citizens' assemblies}, deliberative forums composed of a random sample of citizens, are quickly gaining usage around the world~\citep{OECD20} and proudly point to their random selection\emdash often carried out using complex algorithms from computer science~\citep{FGG+21}\emdash as a source of legitimacy.
If this trend continues, randomness will be associated by the public with neutrality and fairness, not with haphazardness.
Hence, randomized apportionment methods (though, perhaps, simpler ones than the ones we developed here) might yet receive serious consideration.

\subsection*{Acknowledgements}
We thank Bailey Flanigan, Hadi Hosseini, David Wajc, and Peyton Young for helpful discussions.
This work was partially supported by the National Science Foundation under grants IIS-2147187, IIS-2229881 and CCF-2007080 and by the Office of Naval Research under grant N00014-20-1-2488.
Part of this work was done while P.G.\ was visiting the Simons Institute for the Theory of Computing; he gratefully acknowledges the NSF's support of FODSI through grant DMS-2023505.

\bibliographystyle{ACM-Reference-Format}
%\bibliography{random-apportionment}

\begin{thebibliography}{71}
	
	%%% the Web address.
	
	\ifx \showCODEN    \undefined \def \showCODEN     #1{\unskip}     \fi
	\ifx \showDOI      \undefined \def \showDOI       #1{#1}\fi
	\ifx \showISBNx    \undefined \def \showISBNx     #1{\unskip}     \fi
	\ifx \showISBNxiii \undefined \def \showISBNxiii  #1{\unskip}     \fi
	\ifx \showISSN     \undefined \def \showISSN      #1{\unskip}     \fi
	\ifx \showLCCN     \undefined \def \showLCCN      #1{\unskip}     \fi
	\ifx \shownote     \undefined \def \shownote      #1{#1}          \fi
	\ifx \showarticletitle \undefined \def \showarticletitle #1{#1}   \fi
	\ifx \showURL      \undefined \def \showURL       {\relax}        \fi
	% invisible to TeX
	\providecommand\bibfield[2]{#2}
	\providecommand\bibinfo[2]{#2}
	\providecommand\natexlab[1]{#1}
	\providecommand\showeprint[2][]{arXiv:#2}
	
	\bibitem[Ageev and Sviridenko(2004)]%
	{AS04}
	\bibfield{author}{\bibinfo{person}{A.~A. Ageev} {and} \bibinfo{person}{M.~I.
			Sviridenko}.} \bibinfo{year}{2004}\natexlab{}.
	\newblock \showarticletitle{Pipage Rounding: {{A}} New Method of Constructing
		Algorithms with Proven Performance Guarantee}.
	\newblock \bibinfo{journal}{\emph{Journal of Combinatorial Optimization}}
	\bibinfo{volume}{8}, \bibinfo{number}{3} (\bibinfo{year}{2004}),
	\bibinfo{pages}{307--328}.
	\newblock
	
	
	\bibitem[Agnew(2008)]%
	{Agnew08}
	\bibfield{author}{\bibinfo{person}{R.~A. Agnew}.}
	\bibinfo{year}{2008}\natexlab{}.
	\newblock \showarticletitle{Optimal Congressional Apportionment}.
	\newblock \bibinfo{journal}{\emph{The American Mathematical Monthly}}
	\bibinfo{volume}{115}, \bibinfo{number}{4} (\bibinfo{year}{2008}),
	\bibinfo{pages}{297--303}.
	\newblock
	
	
	\bibitem[Akbarpour and Nikzad(2020)]%
	{AN20}
	\bibfield{author}{\bibinfo{person}{M. Akbarpour} {and} \bibinfo{person}{A.
			Nikzad}.} \bibinfo{year}{2020}\natexlab{}.
	\newblock \showarticletitle{Approximate Random Allocation Mechanisms}.
	\newblock \bibinfo{journal}{\emph{The Review of Economic Studies}}
	\bibinfo{volume}{87}, \bibinfo{number}{6} (\bibinfo{year}{2020}),
	\bibinfo{pages}{2473--2510}.
	\newblock
	
	
	\bibitem[Aziz et~al\mbox{.}(2023a)]%
	{AFS+23}
	\bibfield{author}{\bibinfo{person}{H. Aziz}, \bibinfo{person}{R. Freeman},
		\bibinfo{person}{N. Shah}, {and} \bibinfo{person}{R. Vaish}.}
	\bibinfo{year}{2023}\natexlab{a}.
	\newblock \showarticletitle{Best of Both Worlds: Ex Ante and Ex Post Fairness
		in Resource Allocation}.
	\newblock \bibinfo{journal}{\emph{Operations Research}} (\bibinfo{year}{2023}).
	\newblock
	
	
	\bibitem[Aziz et~al\mbox{.}(2023b)]%
	{AGM23}
	\bibfield{author}{\bibinfo{person}{H. Aziz}, \bibinfo{person}{A. Ganguly},
		{and} \bibinfo{person}{E. Micha}.} \bibinfo{year}{2023}\natexlab{b}.
	\newblock \showarticletitle{Best of Both Worlds Fairness under Entitlements}.
	In \bibinfo{booktitle}{\emph{Proceedings of the 2023 {{International
					Conference}} on {{Autonomous Agents}} and {{Multiagent Systems}}
			({{AAMAS}})}}. \bibinfo{pages}{941--948}.
	\newblock
	
	
	\bibitem[Aziz et~al\mbox{.}(2019)]%
	{ALM+19}
	\bibfield{author}{\bibinfo{person}{H. Aziz}, \bibinfo{person}{O. Lev},
		\bibinfo{person}{N. Mattei}, \bibinfo{person}{J.~S. Rosenschein}, {and}
		\bibinfo{person}{T. Walsh}.} \bibinfo{year}{2019}\natexlab{}.
	\newblock \showarticletitle{Strategyproof Peer Selection Using Randomization,
		Partitioning, and Apportionment}.
	\newblock \bibinfo{journal}{\emph{Artificial Intelligence}}
	\bibinfo{volume}{275} (\bibinfo{year}{2019}), \bibinfo{pages}{295--309}.
	\newblock
	
	
	\bibitem[Aziz et~al\mbox{.}(2020)]%
	{AMS20}
	\bibfield{author}{\bibinfo{person}{H. Aziz}, \bibinfo{person}{H. Moulin}, {and}
		\bibinfo{person}{F. Sandomirskiy}.} \bibinfo{year}{2020}\natexlab{}.
	\newblock \showarticletitle{A Polynomial-Time Algorithm for Computing a
		{{Pareto}} Optimal and Almost Proportional Allocation}.
	\newblock \bibinfo{journal}{\emph{Operations Research Letters}}
	\bibinfo{volume}{48}, \bibinfo{number}{5} (\bibinfo{year}{2020}),
	\bibinfo{pages}{573--578}.
	\newblock
	
	
	\bibitem[Babaioff et~al\mbox{.}(2022)]%
	{BEF22}
	\bibfield{author}{\bibinfo{person}{M. Babaioff}, \bibinfo{person}{T. Ezra},
		{and} \bibinfo{person}{U. Feige}.} \bibinfo{year}{2022}\natexlab{}.
	\newblock \showarticletitle{On Best-of-Both-Worlds Fair-Share Allocations}.
	\newblock In \bibinfo{booktitle}{\emph{Web and {{Internet Economics}} (WINE
			2022)}}. \bibinfo{series}{Lecture Notes in Computer Science},
	Vol.~\bibinfo{volume}{13778}. \bibinfo{publisher}{Springer},
	\bibinfo{pages}{237--255}.
	\newblock
	
	
	\bibitem[Balinski(1993)]%
	{Balinski93}
	\bibfield{author}{\bibinfo{person}{M. Balinski}.}
	\bibinfo{year}{1993}\natexlab{}.
	\newblock \showarticletitle{The Problem with Apportionment}.
	\newblock \bibinfo{journal}{\emph{Journal of the Operations Research Society of
			Japan}} \bibinfo{volume}{36}, \bibinfo{number}{3} (\bibinfo{year}{1993}),
	\bibinfo{pages}{134--148}.
	\newblock
	
	
	\bibitem[Balinski and Ram{\'{\i}}rez(1999)]%
	{BR99}
	\bibfield{author}{\bibinfo{person}{M. Balinski} {and} \bibinfo{person}{V.
			Ram{\'{\i}}rez}.} \bibinfo{year}{1999}\natexlab{}.
	\newblock \showarticletitle{Parametric Methods of Apportionment, Rounding and
		Production}.
	\newblock \bibinfo{journal}{\emph{Mathematical Social Sciences}}
	\bibinfo{volume}{37}, \bibinfo{number}{2} (\bibinfo{year}{1999}),
	\bibinfo{pages}{107--122}.
	\newblock
	
	
	\bibitem[Balinski and Ramirez(2014)]%
	{BR14}
	\bibfield{author}{\bibinfo{person}{M. Balinski} {and} \bibinfo{person}{V.
			Ramirez}.} \bibinfo{year}{2014}\natexlab{}.
	\newblock \showarticletitle{Parametric vs. Divisor Methods of Apportionment}.
	\newblock \bibinfo{journal}{\emph{Annals of Operations Research}}
	\bibinfo{volume}{215}, \bibinfo{number}{1} (\bibinfo{year}{2014}),
	\bibinfo{pages}{39--48}.
	\newblock
	
	
	\bibitem[Balinski and Shahidi(1998)]%
	{BS98}
	\bibfield{author}{\bibinfo{person}{M. Balinski} {and} \bibinfo{person}{N.
			Shahidi}.} \bibinfo{year}{1998}\natexlab{}.
	\newblock \showarticletitle{A Simple Approach to the Product Rate Variation
		Problem via Axiomatics}.
	\newblock \bibinfo{journal}{\emph{Operations Research Letters}}
	\bibinfo{volume}{22}, \bibinfo{number}{4-5} (\bibinfo{year}{1998}),
	\bibinfo{pages}{129--135}.
	\newblock
	
	
	\bibitem[Balinski and Demange(1989a)]%
	{BD89a}
	\bibfield{author}{\bibinfo{person}{M.~L. Balinski} {and} \bibinfo{person}{G.
			Demange}.} \bibinfo{year}{1989}\natexlab{a}.
	\newblock \showarticletitle{Algorithms for Proportional Matrices in Reals and
		Integers}.
	\newblock \bibinfo{journal}{\emph{Mathematical Programming}}
	\bibinfo{volume}{45}, \bibinfo{number}{1-3} (\bibinfo{year}{1989}),
	\bibinfo{pages}{193--210}.
	\newblock
	
	
	\bibitem[Balinski and Demange(1989b)]%
	{BD89}
	\bibfield{author}{\bibinfo{person}{M.~L. Balinski} {and} \bibinfo{person}{G.
			Demange}.} \bibinfo{year}{1989}\natexlab{b}.
	\newblock \showarticletitle{An Axiomatic Approach to Proportionality Between
		Matrices}.
	\newblock \bibinfo{journal}{\emph{Mathematics of Operations Research}}
	\bibinfo{volume}{14}, \bibinfo{number}{4} (\bibinfo{year}{1989}),
	\bibinfo{pages}{700--719}.
	\newblock
	
	
	\bibitem[Balinski and Young(1975)]%
	{BY75}
	\bibfield{author}{\bibinfo{person}{M.~L. Balinski} {and} \bibinfo{person}{H.~P.
			Young}.} \bibinfo{year}{1975}\natexlab{}.
	\newblock \showarticletitle{The Quota Method of Apportionment}.
	\newblock \bibinfo{journal}{\emph{The American Mathematical Monthly}}
	\bibinfo{volume}{82}, \bibinfo{number}{7} (\bibinfo{year}{1975}),
	\bibinfo{pages}{701--730}.
	\newblock
	
	
	\bibitem[Balinski and Young(1979)]%
	{BY79}
	\bibfield{author}{\bibinfo{person}{M.~L. Balinski} {and} \bibinfo{person}{H.~P.
			Young}.} \bibinfo{year}{1979}\natexlab{}.
	\newblock \showarticletitle{Quotatone Apportionment Methods}.
	\newblock \bibinfo{journal}{\emph{Mathematics of Operations Research}}
	\bibinfo{volume}{4}, \bibinfo{number}{1} (\bibinfo{year}{1979}),
	\bibinfo{pages}{31--38}.
	\newblock
	
	
	\bibitem[Balinski and Young(1982)]%
	{BY82}
	\bibfield{author}{\bibinfo{person}{M.~L. Balinski} {and} \bibinfo{person}{H.~P.
			Young}.} \bibinfo{year}{1982}\natexlab{}.
	\newblock \bibinfo{booktitle}{\emph{Fair Representation: Meeting the Ideal of
			One Man, One Vote}}.
	\newblock \bibinfo{publisher}{Yale University Press}.
	\newblock
	
	
	\bibitem[Bansal et~al\mbox{.}(2012)]%
	{BGLM+12}
	\bibfield{author}{\bibinfo{person}{N. Bansal}, \bibinfo{person}{A. Gupta},
		\bibinfo{person}{J. Li}, \bibinfo{person}{J. Mestre}, \bibinfo{person}{V.
			Nagarajan}, {and} \bibinfo{person}{A. Rudra}.}
	\bibinfo{year}{2012}\natexlab{}.
	\newblock \showarticletitle{When {LP} Is the Cure for Your Matching Woes:
		Improved Bounds for Stochastic Matchings}.
	\newblock \bibinfo{journal}{\emph{Algorithmica}} \bibinfo{volume}{63},
	\bibinfo{number}{4} (\bibinfo{year}{2012}), \bibinfo{pages}{733--762}.
	\newblock
	
	
	\bibitem[Barbanel(1996)]%
	{Barbanel96}
	\bibfield{author}{\bibinfo{person}{J. Barbanel}.}
	\bibinfo{year}{1996}\natexlab{}.
	\newblock \showarticletitle{Game-Theoretic Algorithms for Fair and Strongly
		Fair Cake Division with Entitlements}. In
	\bibinfo{booktitle}{\emph{Colloquium {{Mathematicae}}}},
	Vol.~\bibinfo{volume}{69}. \bibinfo{pages}{59--73}.
	\newblock
	
	
	\bibitem[Bautista et~al\mbox{.}(1996)]%
	{BCC96}
	\bibfield{author}{\bibinfo{person}{J. Bautista}, \bibinfo{person}{R. Companys},
		{and} \bibinfo{person}{A. Corominas}.} \bibinfo{year}{1996}\natexlab{}.
	\newblock \showarticletitle{A Note on the Relation between the Product Rate
		Variation ({{PRV}}) Problem and the Apportionment Problem}.
	\newblock \bibinfo{journal}{\emph{Journal of the Operational Research Society}}
	\bibinfo{volume}{47}, \bibinfo{number}{11} (\bibinfo{year}{1996}),
	\bibinfo{pages}{1410--1414}.
	\newblock
	
	
	\bibitem[Benad{\`e} et~al\mbox{.}(2023)]%
	{BKP+23}
	\bibfield{author}{\bibinfo{person}{G. Benad{\`e}}, \bibinfo{person}{A.~M.
			Kazachkov}, \bibinfo{person}{A.~D. Procaccia}, \bibinfo{person}{A. Psomas},
		{and} \bibinfo{person}{D. Zeng}.} \bibinfo{year}{2023}\natexlab{}.
	\newblock \showarticletitle{Fair and Efficient Online Allocations}.
	\newblock \bibinfo{journal}{\emph{Operations Research}} (\bibinfo{year}{2023}).
	\newblock
	
	
	\bibitem[Birkhoff(1946)]%
	{Birk46}
	\bibfield{author}{\bibinfo{person}{G. Birkhoff}.}
	\bibinfo{year}{1946}\natexlab{}.
	\newblock \showarticletitle{Three Observations on Linear Algebra}.
	\newblock \bibinfo{journal}{\emph{Universidad Nacional de Tucum\'an, Revista
			A}}  \bibinfo{volume}{5} (\bibinfo{year}{1946}), \bibinfo{pages}{147--151}.
	\newblock
	
	
	\bibitem[Bogomolnaia and Moulin(2001)]%
	{BM01}
	\bibfield{author}{\bibinfo{person}{A. Bogomolnaia} {and} \bibinfo{person}{H.
			Moulin}.} \bibinfo{year}{2001}\natexlab{}.
	\newblock \showarticletitle{A New Solution to the Random Assignment Problem}.
	\newblock \bibinfo{journal}{\emph{Journal of Economic Theory}}
	\bibinfo{volume}{100} (\bibinfo{year}{2001}), \bibinfo{pages}{295--328}.
	\newblock
	
	
	\bibitem[Brewer and Hanif(1983)]%
	{BH83}
	\bibfield{author}{\bibinfo{person}{K.~R.~W. Brewer} {and} \bibinfo{person}{M.
			Hanif}.} \bibinfo{year}{1983}\natexlab{}.
	\newblock \bibinfo{booktitle}{\emph{An {{Introduction}} to {{Sampling}} with
			{{Unequal Probabilities}}}}. \bibinfo{series}{Lecture Notes in Statistics},
	Vol.~\bibinfo{volume}{15}.
	\newblock \bibinfo{publisher}{{Springer}}.
	\newblock
	
	
	\bibitem[Buchstein and Hein(2009)]%
	{BH09}
	\bibfield{author}{\bibinfo{person}{H. Buchstein} {and} \bibinfo{person}{M.
			Hein}.} \bibinfo{year}{2009}\natexlab{}.
	\newblock \showarticletitle{Randomizing {{Europe}}: The Lottery as a
		Decision-Making Procedure for Policy Creation in the {{EU}}}.
	\newblock \bibinfo{journal}{\emph{Critical Policy Studies}}
	\bibinfo{volume}{3}, \bibinfo{number}{1} (\bibinfo{year}{2009}),
	\bibinfo{pages}{29--57}.
	\newblock
	
	
	\bibitem[Buchstein et~al\mbox{.}(2013)]%
	{BHJ13}
	\bibfield{author}{\bibinfo{person}{H. Buchstein}, \bibinfo{person}{M. Hein},
		{and} \bibinfo{person}{J. J{\"u}nger}.} \bibinfo{year}{2013}\natexlab{}.
	\newblock \showarticletitle{Die ,,{{EU}}-{{Kommissionslotterie}}``. {{Eine
				Simulations\-studie}}}. In \bibinfo{booktitle}{\emph{Die {{Versprechen}} der
			{{Demokratie}}}}. \bibinfo{publisher}{{Nomos}}, \bibinfo{pages}{190--227}.
	\newblock
	
	
	\bibitem[Budish et~al\mbox{.}(2013)]%
	{BCKM13}
	\bibfield{author}{\bibinfo{person}{E. Budish}, \bibinfo{person}{Y.-K. Che},
		\bibinfo{person}{F. Kojima}, {and} \bibinfo{person}{P. Milgrom}.}
	\bibinfo{year}{2013}\natexlab{}.
	\newblock \showarticletitle{Designing Random Allocation Mechanisms: Theory and
		Applications}.
	\newblock \bibinfo{journal}{\emph{American Economic Review}}
	\bibinfo{volume}{103}, \bibinfo{number}{2} (\bibinfo{year}{2013}),
	\bibinfo{pages}{585--623}.
	\newblock
	
	
	\bibitem[Burt and Harris(1963)]%
	{BH63}
	\bibfield{author}{\bibinfo{person}{O.~R. Burt} {and} \bibinfo{person}{C.~C.
			Harris}.} \bibinfo{year}{1963}\natexlab{}.
	\newblock \showarticletitle{Letter to the Editor---{{Apportionment}} of the
		{U.S.} {House} of {Representatives}: A Minimum Range, Integer Solution,
		Allocation Problem}.
	\newblock \bibinfo{journal}{\emph{Operations Research}} \bibinfo{volume}{11},
	\bibinfo{number}{4} (\bibinfo{year}{1963}), \bibinfo{pages}{648--652}.
	\newblock
	
	
	\bibitem[Cembrano et~al\mbox{.}(2024)]%
	{CCS+24}
	\bibfield{author}{\bibinfo{person}{J. Cembrano}, \bibinfo{person}{J. Correa},
		\bibinfo{person}{U. Schmidt-Kraepelin}, \bibinfo{person}{A.
			Tsigonias-Dimitriadis}, {and} \bibinfo{person}{V. Verdugo}.}
	\bibinfo{year}{2024}\natexlab{}.
	\newblock \bibinfo{title}{New Combinatorial Insights for Monotone
		Apportionment}.  (\bibinfo{year}{2024}).
	\newblock
	
	
	\bibitem[Cembrano et~al\mbox{.}(2022)]%
	{CCV22}
	\bibfield{author}{\bibinfo{person}{J. Cembrano}, \bibinfo{person}{J. Correa},
		{and} \bibinfo{person}{V. Verdugo}.} \bibinfo{year}{2022}\natexlab{}.
	\newblock \showarticletitle{Multidimensional Political Apportionment}.
	\newblock \bibinfo{journal}{\emph{Proceedings of the National Academy of
			Sciences}} \bibinfo{volume}{119}, \bibinfo{number}{15}
	(\bibinfo{year}{2022}), \bibinfo{pages}{e2109305119}.
	\newblock
	
	
	\bibitem[Chakraborty et~al\mbox{.}(2021)]%
	{CSS21}
	\bibfield{author}{\bibinfo{person}{M. Chakraborty}, \bibinfo{person}{U.
			{Schmidt-Kraepelin}}, {and} \bibinfo{person}{W. Suksompong}.}
	\bibinfo{year}{2021}\natexlab{}.
	\newblock \showarticletitle{Picking Sequences and Monotonicity in Weighted Fair
		Division}.
	\newblock \bibinfo{journal}{\emph{Artificial Intelligence}}
	\bibinfo{volume}{301} (\bibinfo{year}{2021}), \bibinfo{pages}{103578}.
	\newblock
	
	
	\bibitem[Cheng et~al\mbox{.}(2019)]%
	{CJMW19}
	\bibfield{author}{\bibinfo{person}{Y. Cheng}, \bibinfo{person}{Z. Jiang},
		\bibinfo{person}{K. Munagala}, {and} \bibinfo{person}{K. Wang}.}
	\bibinfo{year}{2019}\natexlab{}.
	\newblock \showarticletitle{Group Fairness in Committee Selection}. In
	\bibinfo{booktitle}{\emph{Proceedings of the 20th ACM Conference on Economics
			and Computation (EC)}}. \bibinfo{pages}{263--279}.
	\newblock
	
	
	\bibitem[Correa et~al\mbox{.}(2024)]%
	{CGS+24}
	\bibfield{author}{\bibinfo{person}{J. Correa}, \bibinfo{person}{P. G{\"o}lz},
		\bibinfo{person}{U. {Schmidt-Kraepelin}}, \bibinfo{person}{J.
			{Tucker-Foltz}}, {and} \bibinfo{person}{V. Verdugo}.}
	\bibinfo{year}{2024}\natexlab{}.
	\newblock \showarticletitle{Monotone randomized apportionment}. In
	\bibinfo{booktitle}{\emph{Proceedings of the 25th ACM Conference on Economics
			and Computation (EC)}}.
	\newblock
	\newblock
	\shownote{Forthcoming}.
	
	
	\bibitem[Deville and Tille(1998)]%
	{DT98}
	\bibfield{author}{\bibinfo{person}{J.-C. Deville} {and} \bibinfo{person}{Y.
			Tille}.} \bibinfo{year}{1998}\natexlab{}.
	\newblock \showarticletitle{Unequal Probability Sampling without Replacement
		through a Splitting Method}.
	\newblock \bibinfo{journal}{\emph{Biometrika}} \bibinfo{volume}{85},
	\bibinfo{number}{1} (\bibinfo{year}{1998}), \bibinfo{pages}{89--101}.
	\newblock
	
	
	\bibitem[{El-Helaly}(2019)]%
	{El-Helaly19a}
	\bibfield{author}{\bibinfo{person}{S. {El-Helaly}}.}
	\bibinfo{year}{2019}\natexlab{}.
	\newblock \bibinfo{booktitle}{\emph{The Mathematics of Voting and
			Apportionment: An Introduction}}.
	\newblock \bibinfo{publisher}{{Springer}}.
	\newblock
	
	
	\bibitem[Elliott et~al\mbox{.}(2021)]%
	{EMS+21}
	\bibfield{author}{\bibinfo{person}{D. Elliott}, \bibinfo{person}{S. Martin},
		\bibinfo{person}{J. Shakesprere}, {and} \bibinfo{person}{J. Kelly}.}
	\bibinfo{year}{2021}\natexlab{}.
	\newblock \bibinfo{booktitle}{\emph{Simulating the 2020 Census: Miscounts and
			the Fairness of Outcomes}}.
	\newblock \bibinfo{type}{Research {{Report}}}. \bibinfo{institution}{{Urban
			Institute}}.
	\newblock
	
	
	\bibitem[Ernst(1994)]%
	{Ernst94}
	\bibfield{author}{\bibinfo{person}{L.~R. Ernst}.}
	\bibinfo{year}{1994}\natexlab{}.
	\newblock \showarticletitle{Apportionment Methods for the {{House}} of
		{{Representatives}} and the Court Challenges}.
	\newblock \bibinfo{journal}{\emph{Management Science}} \bibinfo{volume}{40},
	\bibinfo{number}{10} (\bibinfo{year}{1994}), \bibinfo{pages}{1207--1227}.
	\newblock
	
	
	\bibitem[Evren and Khanna(2024)]%
	{EK24}
	\bibfield{author}{\bibinfo{person}{H. Evren} {and} \bibinfo{person}{M.
			Khanna}.} \bibinfo{year}{2024}\natexlab{}.
	\newblock \bibinfo{title}{Affirmative Action's Cumulative Fractional
		Assignments}.
	\newblock
	\newblock
	\showeprint[arxiv]{2111.11963}~[econ.TH]
	
	
	\bibitem[Feldman et~al\mbox{.}(2023)]%
	{FMN+23}
	\bibfield{author}{\bibinfo{person}{M. Feldman}, \bibinfo{person}{S. Mauras},
		\bibinfo{person}{V.~V. Narayan}, {and} \bibinfo{person}{T. Ponitka}.}
	\bibinfo{year}{2023}\natexlab{}.
	\newblock \bibinfo{title}{Breaking the Envy Cycle: Best-of-Both-Worlds
		Guarantees for Subadditive Valuations}.
	\newblock
	\newblock
	\showeprint[arxiv]{2304.03706}~[cs.GT]
	
	
	\bibitem[Flanigan et~al\mbox{.}(2021)]%
	{FGG+21}
	\bibfield{author}{\bibinfo{person}{B. Flanigan}, \bibinfo{person}{P. G{\"o}lz},
		\bibinfo{person}{A. Gupta}, \bibinfo{person}{B. Hennig}, {and}
		\bibinfo{person}{A.~D. Procaccia}.} \bibinfo{year}{2021}\natexlab{}.
	\newblock \showarticletitle{Fair Algorithms for Selecting Citizens'
		Assemblies}.
	\newblock \bibinfo{journal}{\emph{Nature}} \bibinfo{volume}{596},
	\bibinfo{number}{7873} (\bibinfo{year}{2021}), \bibinfo{pages}{548--552}.
	\newblock
	
	
	\bibitem[Gandhi et~al\mbox{.}(2006)]%
	{GKP+06}
	\bibfield{author}{\bibinfo{person}{R. Gandhi}, \bibinfo{person}{S. Khuller},
		\bibinfo{person}{S. Parthasarathy}, {and} \bibinfo{person}{A. Srinivasan}.}
	\bibinfo{year}{2006}\natexlab{}.
	\newblock \showarticletitle{Dependent Rounding and Its Applications to
		Approximation Algorithms}.
	\newblock \bibinfo{journal}{\emph{{Journal of the ACM}}} \bibinfo{volume}{53},
	\bibinfo{number}{3} (\bibinfo{year}{2006}), \bibinfo{pages}{324--360}.
	\newblock
	
	
	\bibitem[Goldman and Procaccia(2014)]%
	{GP14}
	\bibfield{author}{\bibinfo{person}{J. Goldman} {and} \bibinfo{person}{A.~D.
			Procaccia}.} \bibinfo{year}{2014}\natexlab{}.
	\newblock \showarticletitle{Spliddit: Unleashing Fair Division Algorithms}.
	\newblock \bibinfo{journal}{\emph{SIGecom Exchanges}} \bibinfo{volume}{13},
	\bibinfo{number}{2} (\bibinfo{year}{2014}), \bibinfo{pages}{41--46}.
	\newblock
	
	
	\bibitem[Grimmett(2004)]%
	{Grim04}
	\bibfield{author}{\bibinfo{person}{G. Grimmett}.}
	\bibinfo{year}{2004}\natexlab{}.
	\newblock \showarticletitle{Stochastic Apportionment}.
	\newblock \bibinfo{journal}{\emph{{American Mathematical Monthly}}}
	\bibinfo{volume}{11}, \bibinfo{number}{4} (\bibinfo{year}{2004}),
	\bibinfo{pages}{299--307}.
	\newblock
	
	
	\bibitem[Hall(1935)]%
	{Hall35}
	\bibfield{author}{\bibinfo{person}{P. Hall}.} \bibinfo{year}{1935}\natexlab{}.
	\newblock \showarticletitle{On Representatives of Subsets}.
	\newblock \bibinfo{journal}{\emph{Journal of the London Mathematical Society}}
	\bibinfo{volume}{s1-10}, \bibinfo{number}{1} (\bibinfo{year}{1935}),
	\bibinfo{pages}{26--30}.
	\newblock
	
	
	\bibitem[Hoefer et~al\mbox{.}(2023)]%
	{HSV23}
	\bibfield{author}{\bibinfo{person}{M. Hoefer}, \bibinfo{person}{M.
			Schmalhofer}, {and} \bibinfo{person}{G. Varricchio}.}
	\bibinfo{year}{2023}\natexlab{}.
	\newblock \showarticletitle{Best of Both Worlds: Agents with Entitlements}. In
	\bibinfo{booktitle}{\emph{Proceedings of the 2023 {{International
					Conference}} on {{Autonomous Agents}} and {{Multiagent Systems}}
			({{AAMAS}})}}. \bibinfo{pages}{564--572}.
	\newblock
	
	
	\bibitem[Hong et~al\mbox{.}(2023)]%
	{HNR+23}
	\bibfield{author}{\bibinfo{person}{J.-I. Hong}, \bibinfo{person}{J. Najnudel},
		\bibinfo{person}{S.-M. Rao}, {and} \bibinfo{person}{J.-Y. Yen}.}
	\bibinfo{year}{2023}\natexlab{}.
	\newblock \showarticletitle{Random Apportionment: A Stochastic Solution to the
		Balinski--Young Impossibility}.
	\newblock \bibinfo{journal}{\emph{Methodology and Computing in Applied
			Probability}} \bibinfo{volume}{25}, \bibinfo{number}{4}
	(\bibinfo{year}{2023}), \bibinfo{pages}{91}.
	\newblock
	
	
	\bibitem[Huntington(1928)]%
	{Huntington28}
	\bibfield{author}{\bibinfo{person}{E.~V. Huntington}.}
	\bibinfo{year}{1928}\natexlab{}.
	\newblock \showarticletitle{The Apportionment of Representatives in
		{{Congress}}}.
	\newblock \bibinfo{journal}{\emph{Trans. Amer. Math. Soc.}}
	\bibinfo{volume}{30}, \bibinfo{number}{1} (\bibinfo{year}{1928}),
	\bibinfo{pages}{85--110}.
	\newblock
	
	
	\bibitem[Hylland and Zeckhauser(1979)]%
	{HZ79}
	\bibfield{author}{\bibinfo{person}{A. Hylland} {and} \bibinfo{person}{R.
			Zeckhauser}.} \bibinfo{year}{1979}\natexlab{}.
	\newblock \showarticletitle{The Efficient Allocation of Individuals to
		Positions}.
	\newblock \bibinfo{journal}{\emph{Journal of Political Economy}}
	\bibinfo{volume}{87}, \bibinfo{number}{2} (\bibinfo{year}{1979}),
	\bibinfo{pages}{293--314}.
	\newblock
	
	
	\bibitem[Janson(2014)]%
	{Janson14}
	\bibfield{author}{\bibinfo{person}{S. Janson}.}
	\bibinfo{year}{2014}\natexlab{}.
	\newblock \showarticletitle{Asymptotic Bias of Some Election Methods}.
	\newblock \bibinfo{journal}{\emph{Annals of Operations Research}}
	\bibinfo{volume}{215}, \bibinfo{number}{1} (\bibinfo{year}{2014}),
	\bibinfo{pages}{89--136}.
	\newblock
	
	
	\bibitem[Kingman(1993)]%
	{Kingman93}
	\bibfield{author}{\bibinfo{person}{J.~F.~C. Kingman}.}
	\bibinfo{year}{1993}\natexlab{}.
	\newblock \bibinfo{booktitle}{\emph{Poisson Processes}}.
	\newblock \bibinfo{publisher}{Clarendon Press}.
	\newblock
	
	
	\bibitem[Kubiak(1993)]%
	{Kubiak93}
	\bibfield{author}{\bibinfo{person}{W. Kubiak}.}
	\bibinfo{year}{1993}\natexlab{}.
	\newblock \showarticletitle{Minimizing Variation of Production Rates in
		Just-in-Time Systems: {{A}} Survey}.
	\newblock \bibinfo{journal}{\emph{European Journal of Operational Research}}
	\bibinfo{volume}{66}, \bibinfo{number}{3} (\bibinfo{year}{1993}),
	\bibinfo{pages}{259--271}.
	\newblock
	
	
	\bibitem[Kumar et~al\mbox{.}(2009)]%
	{KMP+09}
	\bibfield{author}{\bibinfo{person}{V.~A. Kumar}, \bibinfo{person}{M.~V.
			Marathe}, \bibinfo{person}{S. Parthasarathy}, {and} \bibinfo{person}{A.
			Srinivasan}.} \bibinfo{year}{2009}\natexlab{}.
	\newblock \showarticletitle{A Unified Approach to Scheduling on Unrelated
		Parallel Machines}.
	\newblock \bibinfo{journal}{\emph{{Journal of the ACM}}} \bibinfo{volume}{56},
	\bibinfo{number}{5} (\bibinfo{year}{2009}), \bibinfo{pages}{1--31}.
	\newblock
	
	
	\bibitem[Kurokawa et~al\mbox{.}(2018)]%
	{KPS18}
	\bibfield{author}{\bibinfo{person}{D. Kurokawa}, \bibinfo{person}{A.~D.
			Procaccia}, {and} \bibinfo{person}{N. Shah}.}
	\bibinfo{year}{2018}\natexlab{}.
	\newblock \showarticletitle{Leximin Allocations in the Real World}.
	\newblock \bibinfo{journal}{\emph{ACM Transactions on Economics and
			Computation}} \bibinfo{volume}{6}, \bibinfo{number}{3--4}
	(\bibinfo{year}{2018}), \bibinfo{pages}{\ article 11}.
	\newblock
	
	
	\bibitem[Lauwers and Van~Puyenbroeck(2006)]%
	{LV06}
	\bibfield{author}{\bibinfo{person}{L. Lauwers} {and} \bibinfo{person}{T.
			Van~Puyenbroeck}.} \bibinfo{year}{2006}\natexlab{}.
	\newblock \showarticletitle{The {Hamilton} Apportionment Method Is Between the
		{Adams} Method and the {Jefferson} Method}.
	\newblock \bibinfo{journal}{\emph{Mathematics of Operations Research}}
	\bibinfo{volume}{31}, \bibinfo{number}{2} (\bibinfo{year}{2006}),
	\bibinfo{pages}{390--397}.
	\newblock
	
	
	\bibitem[Maier et~al\mbox{.}(2010)]%
	{MZZ10}
	\bibfield{author}{\bibinfo{person}{S. Maier}, \bibinfo{person}{P.
			Zachariassen}, {and} \bibinfo{person}{M. Zachariasen}.}
	\bibinfo{year}{2010}\natexlab{}.
	\newblock \showarticletitle{Divisor-Based Biproportional Apportionment in
		Electoral Systems: A Real-Life Benchmark Study}.
	\newblock \bibinfo{journal}{\emph{Management Science}} \bibinfo{volume}{56},
	\bibinfo{number}{2} (\bibinfo{year}{2010}), \bibinfo{pages}{373--387}.
	\newblock
	
	
	\bibitem[Marshall et~al\mbox{.}(2002)]%
	{MOP02}
	\bibfield{author}{\bibinfo{person}{A.~W. Marshall}, \bibinfo{person}{I. Olkin},
		{and} \bibinfo{person}{F. Pukelsheim}.} \bibinfo{year}{2002}\natexlab{}.
	\newblock \showarticletitle{A Majorization Comparison of Apportionment Methods
		in Proportional Representation}.
	\newblock \bibinfo{journal}{\emph{Social Choice and Welfare}}
	\bibinfo{volume}{19}, \bibinfo{number}{4} (\bibinfo{year}{2002}),
	\bibinfo{pages}{885--900}.
	\newblock
	
	
	\bibitem[Mathieu and Verdugo(2022)]%
	{MV22}
	\bibfield{author}{\bibinfo{person}{C. Mathieu} {and} \bibinfo{person}{V.
			Verdugo}.} \bibinfo{year}{2022}\natexlab{}.
	\newblock \showarticletitle{Apportionment with Parity Constraints}.
	\newblock \bibinfo{journal}{\emph{Mathematical Programming}}
	\bibinfo{volume}{203} (\bibinfo{year}{2022}), \bibinfo{pages}{135--168}.
	\newblock
	
	
	\bibitem[Nisan and Ronen(2001)]%
	{NR01}
	\bibfield{author}{\bibinfo{person}{N. Nisan} {and} \bibinfo{person}{A. Ronen}.}
	\bibinfo{year}{2001}\natexlab{}.
	\newblock \showarticletitle{Algorithmic Mechanism Design}.
	\newblock \bibinfo{journal}{\emph{Games and Economic Behavior}}
	\bibinfo{volume}{35}, \bibinfo{number}{1--2} (\bibinfo{year}{2001}),
	\bibinfo{pages}{166--196}.
	\newblock
	
	
	\bibitem[{OECD}(2020)]%
	{OECD20}
	\bibfield{author}{\bibinfo{person}{{OECD}}.} \bibinfo{year}{2020}\natexlab{}.
	\newblock \bibinfo{booktitle}{\emph{Innovative Citizen Participation and New
			Democratic Institutions: Catching the Deliberative Wave}}.
	\newblock \bibinfo{publisher}{{OECD}}.
	\newblock
	
	
	\bibitem[Palomares et~al\mbox{.}(2024)]%
	{PPR24}
	\bibfield{author}{\bibinfo{person}{A. Palomares}, \bibinfo{person}{F.
			Pukelsheim}, {and} \bibinfo{person}{V. Ram{\'i}rez}.}
	\bibinfo{year}{2024}\natexlab{}.
	\newblock \showarticletitle{Note on Axiomatic Properties of Apportionment
		Methods for Proportional Representation Systems}.
	\newblock \bibinfo{journal}{\emph{Mathematical Programming}}
	\bibinfo{volume}{203}, \bibinfo{number}{1-2} (\bibinfo{year}{2024}),
	\bibinfo{pages}{169--185}.
	\newblock
	
	
	\bibitem[Panconesi and Srinivasan(1997)]%
	{PS97}
	\bibfield{author}{\bibinfo{person}{A. Panconesi} {and} \bibinfo{person}{A.
			Srinivasan}.} \bibinfo{year}{1997}\natexlab{}.
	\newblock \showarticletitle{Randomized Distributed Edge Coloring via an
		Extension of the {{Chernoff}}\textendash{{Hoeffding}} Bounds}.
	\newblock \bibinfo{journal}{\emph{{SIAM Journal on Computing}}}
	\bibinfo{volume}{26}, \bibinfo{number}{2} (\bibinfo{year}{1997}),
	\bibinfo{pages}{350--368}.
	\newblock
	
	
	\bibitem[P{\'o}lya(1919)]%
	{Polya19}
	\bibfield{author}{\bibinfo{person}{G. P{\'o}lya}.}
	\bibinfo{year}{1919}\natexlab{}.
	\newblock \showarticletitle{{Proportionalwahl und
			Wahrscheinlichkeitsrechnung}}.
	\newblock \bibinfo{journal}{\emph{Zeitschrift f{\"u}r die gesamte
			Staatswissenschaft}} \bibinfo{volume}{74}, \bibinfo{number}{3}
	(\bibinfo{year}{1919}), \bibinfo{pages}{297--322}.
	\newblock
	
	
	\bibitem[Pukelsheim(2017)]%
	{Pukelsheim17}
	\bibfield{author}{\bibinfo{person}{F. Pukelsheim}.}
	\bibinfo{year}{2017}\natexlab{}.
	\newblock \bibinfo{booktitle}{\emph{Proportional Representation: Apportionment
			Methods and Their Applications} (\bibinfo{edition}{2nd} ed.)}.
	\newblock \bibinfo{publisher}{Springer}.
	\newblock
	
	
	\bibitem[Robinson and Ullman(2010)]%
	{RU10}
	\bibfield{author}{\bibinfo{person}{E.~A. Robinson} {and} \bibinfo{person}{D.
			Ullman}.} \bibinfo{year}{2010}\natexlab{}.
	\newblock \bibinfo{booktitle}{\emph{A Mathematical Look at Politics}}.
	\newblock \bibinfo{publisher}{{CRC Press}}.
	\newblock
	
	
	\bibitem[Saha and Srinivasan(2018)]%
	{SS18a}
	\bibfield{author}{\bibinfo{person}{B. Saha} {and} \bibinfo{person}{A.
			Srinivasan}.} \bibinfo{year}{2018}\natexlab{}.
	\newblock \showarticletitle{A New Approximation Technique for
		Resource-Allocation Problems}.
	\newblock \bibinfo{journal}{\emph{Random Structures \& Algorithms}}
	\bibinfo{volume}{52}, \bibinfo{number}{4} (\bibinfo{year}{2018}),
	\bibinfo{pages}{680--715}.
	\newblock
	
	
	\bibitem[Sandomirskiy and {Segal-Halevi}(2022)]%
	{SS22}
	\bibfield{author}{\bibinfo{person}{F. Sandomirskiy} {and} \bibinfo{person}{E.
			{Segal-Halevi}}.} \bibinfo{year}{2022}\natexlab{}.
	\newblock \showarticletitle{Efficient Fair Division with Minimal Sharing}.
	\newblock \bibinfo{journal}{\emph{Operations Research}} \bibinfo{volume}{70},
	\bibinfo{number}{3} (\bibinfo{year}{2022}), \bibinfo{pages}{1762--1782}.
	\newblock
	
	
	\bibitem[Steiner and Yeomans(1993)]%
	{SY93}
	\bibfield{author}{\bibinfo{person}{G. Steiner} {and} \bibinfo{person}{S.
			Yeomans}.} \bibinfo{year}{1993}\natexlab{}.
	\newblock \showarticletitle{Level Schedules for Mixed-Model, Just-in-Time
		Processes}.
	\newblock \bibinfo{journal}{\emph{Management Science}} \bibinfo{volume}{39},
	\bibinfo{number}{6} (\bibinfo{year}{1993}), \bibinfo{pages}{728--735}.
	\newblock
	
	
	\bibitem[Still(1979)]%
	{Still79}
	\bibfield{author}{\bibinfo{person}{J.~W. Still}.}
	\bibinfo{year}{1979}\natexlab{}.
	\newblock \showarticletitle{A Class of New Methods for Congressional
		Apportionment}.
	\newblock \bibinfo{journal}{\emph{{SIAM Journal on Applied Mathematics}}}
	\bibinfo{volume}{37}, \bibinfo{number}{2} (\bibinfo{year}{1979}),
	\bibinfo{pages}{401--418}.
	\newblock
	
	
	\bibitem[Stone(2011)]%
	{Stone11}
	\bibfield{author}{\bibinfo{person}{P. Stone}.} \bibinfo{year}{2011}\natexlab{}.
	\newblock \bibinfo{booktitle}{\emph{The Luck of the Draw: {{The}} Role of
			Lotteries in Decision Making}}.
	\newblock \bibinfo{publisher}{{Oxford University Press}}.
	\newblock
	
	
	\bibitem[Szpiro(2010)]%
	{Szpir10}
	\bibfield{author}{\bibinfo{person}{G.~G. Szpiro}.}
	\bibinfo{year}{2010}\natexlab{}.
	\newblock \bibinfo{booktitle}{\emph{Numbers Rule: The Vexing Mathematics of
			Democracy, from Plato to the Present}}.
	\newblock \bibinfo{publisher}{Princeton University Press}.
	\newblock
	
	
	\bibitem[von Neumann(1953)]%
	{Neumann53}
	\bibfield{author}{\bibinfo{person}{J. von Neumann}.}
	\bibinfo{year}{1953}\natexlab{}.
	\newblock \showarticletitle{A Certain Zero-Sum Two-Person Game Equivalent to
		the Optimal Assignment Problem}.
	\newblock In \bibinfo{booktitle}{\emph{Contributions to the Theory of Games}},
	\bibfield{editor}{\bibinfo{person}{W.~Kuhn} {and} \bibinfo{person}{A.~W.
			Tucker}} (Eds.). Vol.~\bibinfo{volume}{2}. \bibinfo{publisher}{Princeton
		University Press}, \bibinfo{pages}{5--12}.
	\newblock
	
	
\end{thebibliography}

\newpage
\appendix

\section*{\LARGE Proofs of Statements}

\section{Pitfalls in the Development of House Monotone Methods}
\label{app:pitfalls}
In \cref{sec:pitfalls}, we claimed that for the population profile $\vec{p} = (45, 25, 15, 15)$ and house size $h=3$\emdash{}thus, the standard quotas $(1.35, 0.75, 0.45, 0.45)$\emdash{}the following distribution over apportionments can be part of an apportionment method in which house monotonicity, quota, and ex ante proportionality hold across the inputs $\{(\vec{p}, h') \mid h' \leq 3\}$:
\[ \vec{a} = \left\{\begin{array}{@{}ll@{\qquad}ll@{}}
(2, 1, 0, 0) & \text{with probability $35\%$,} &  (1, 1, 0, 1) & \text{with probability $20\%$, and} \\
(1, 1, 1, 0) & \text{with probability $20\%$,} & (1, 0, 1, 1) & \text{with probability $25\%$.}
\end{array}\right.
\]
To obtain such an apportionment method, consider the following capacitated flow network:
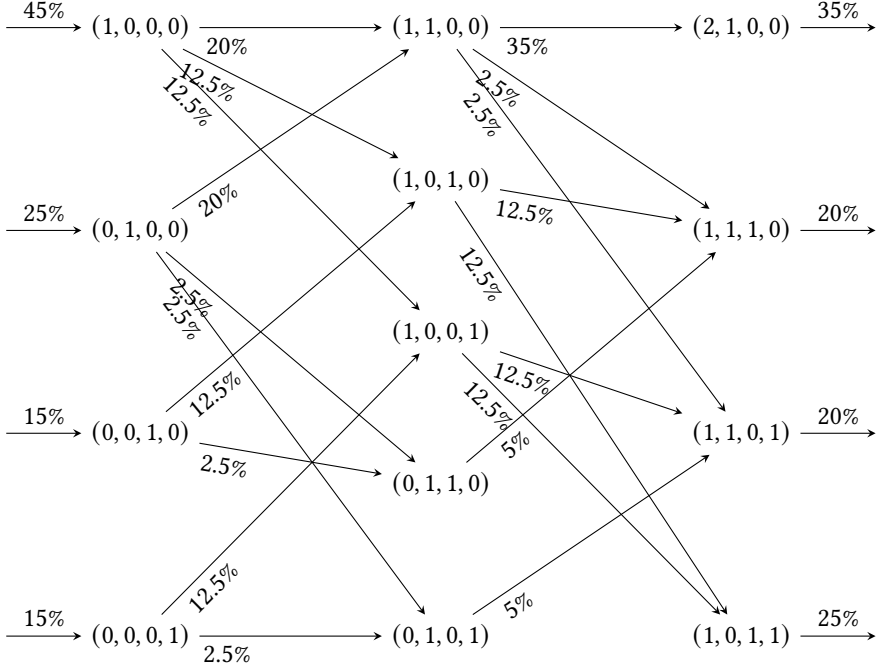
\begin{center}
\begin{tikzpicture}%
    \foreach \row/\col/\sone/\stwo/\sthree/\sfour in {1/1/1/0/0/0,2/1/0/1/0/0,3/1/0/0/1/0,4/1/0/0/0/1,1/2/1/1/0/0,2/2/1/0/1/0,3/2/1/0/0/1,4/2/0/1/1/0,5/2/0/1/0/1,1/3/2/1/0/0,2/3/1/1/1/0,3/3/1/1/0/1,4/3/1/0/1/1} {
    \node at (4*\col, {\if\col2 -2*(\row - 1) \else -(8/3)*(\row-1) \fi}) (s\row\col) {$(\sone, \stwo, \sthree, \sfour)$};
    }
    \foreach \row/\perc in {1/45,2/25,3/15,4/15} {
    \node [left=of s\row1] (s\row0) {};
    \draw (s\row0) [-stealth]  -- node [above] {\perc\%} (s\row1);
    }
    \foreach \row/\perc in {1/35,2/20,3/20,4/25} {
        \node [right=of s\row3] (s\row4) {};
    \draw (s\row3)  [-stealth] -- node [above] {\perc\%} (s\row4);
    }
    \foreach \fromcol/\tocol/\fromrow/\torow/\perc in {1/2/1/1/20,1/2/1/2/12.5,1/2/1/3/12.5,1/2/2/1/20,1/2/2/4/2.5,1/2/2/5/2.5,1/2/3/2/12.5,1/2/3/4/2.5,1/2/4/3/12.5,1/2/4/5/2.5,2/3/1/1/35,2/3/1/2/2.5,2/3/1/3/2.5,2/3/2/2/12.5,2/3/2/4/12.5,2/3/3/3/12.5,2/3/3/4/12.5,2/3/4/2/5,2/3/5/3/5} {
     \draw (s\fromrow\fromcol) [-stealth] -- node [pos=0.15, sloped,below] {\perc\%} (s\torow\tocol);
    }
\end{tikzpicture}
\end{center}
One easily verifies that it is possible to send a total flow of 1 through this network, which necessarily uses all edges at their capacity.
Consider any decomposition of this flow into paths. Our method will be defined as (1) choosing one of these paths $\vec{a}_1 \to \vec{a}_2 \to \vec{a}_3$ with probability proportional to its amount of flow and (2) returning a solution $f$ such that $f(\vec{p},j)= \vec{a}_j$ for $j=1,2,3$, and some canonical apportionment for all other inputs.

Across the inputs $\{(\vec{p}, h') \mid h' \leq 3\}$, this method does not violate house monotonicity since edges in the flow network are such that no agent's seat number decreases along an edge.
Since one verifies that all apportionments labeling the nodes of the flow network satisfy quota, the method satisfies quota on $\{(\vec{p}, h') \mid h' \leq 3\}$.
On the same set of inputs, the method satisfies ex-ante proportionality, which follows from the fact that, for $h'=1,2,3$, weighting the apportionments of the $h'$-th layer of the flow network by their internal flow, we obtain the vector of standard quotas for $\vec{p}$ and house size $h'$.
The egress edges also ensure that, indeed, the method's distribution over apportionments on $\vec{p}, h$ is as given above.

Finally, we must prove that none of the apportionments in the last layer of the flow network are toxic.
For this, observe that the quota solution by \citet{BY75} (which is house monotone and satisfies quota), for one way of breaking ties in the definition, produces the following apportionments on $\vec{p}$:
$(1,0,0,0) \to (1, 1, 0, 0) \to (2, 1, 0, 0) \to (2, 1, 1, 0) \to (2, 1, 1, 1), \dots$.
Since $(2, 1, 0, 0)$ coincides with one of these values, it can be extended by the suffix of the quota solution and is therefore not toxic.
Furthermore, we can extend $(1, 1, 1, 0) \to (2, 1, 1, 0)$, and then continue as the quota solution; $(1, 1, 0, 1) \to (2, 1, 0, 1) \to (2, 1, 1, 1)$, and then as the quota solution; and $(1, 0, 1, 1) \to (1, 1, 1, 1) \to (2, 1, 1, 1)$, and then as the quota solution.
The claim follows by verifying that these extensions do not violate quota until merging with the quota solution.

\section{Deferred Proofs for House Monotone Apportionment}
\label{app:housemonoproofs}

\lemrepetition*

\begin{proof}
\noindent\emph{``$\Rightarrow$'':} Fix $\alpha$ and some $k \in \natszero$. We must show that the finite seat sequence $\beta^{k+1} \coloneqq \alpha_{k p + 1}, \alpha_{k p + 2} \dots, \alpha_{(k+1) p}$ satisfies quota.
Indeed, for any $1 \leq r \leq p$, the number of seats allocated by $\beta^{k+1}$ to state $i$ at house size $r$ is
\begin{align*}
\left| \{ 1 \leq h' \leq r \mid \beta^{k+1}_{h'} = i \} \right| &= \left| \{ 1 \leq h' \leq r \mid \alpha_{k p + h'} = i \} \right| \\
&= a_i(k p + r) - a_i(k p) \\
&\in \{ \lfloor (k p + r) \, p_i / p \rfloor, \lceil (k p + r) \, p_i / p \rceil\} - k \, p_i \tag*{(since $\alpha$ satisfies quota)} \\
&= \{ \lfloor r \, p_i / p \rfloor + k \, p_i, \lceil r \, p_i / p \rceil + k \, p_i \} - k \, p_i \\
&= \{ \lfloor r \, p_i / p \rfloor, \lceil r \, p_i / p \rceil\},
\end{align*}
which shows that $\alpha$ can be decomposed into finite seat sequences $\{\beta^{k}\}_{k \in \natsone}$ satisfying quota. \medskip

\noindent\emph{``$\Leftarrow$'':} Fix some $h$ and choose $k \coloneqq \lfloor (h - 1)/p \rfloor + 1$ and $r \coloneqq ((h-1) \mathop{\text{mod}} p) + 1$ such that $h = (k - 1) \, p + r$, $k \geq 1$ and $1 \leq r \leq p$.
We will show that $\alpha$'s allocation $a(h)$ on $h$ satisfies quota.
Denoting $\beta^k$'s allocation for a house size $h'$ by $b^k(h')$, it holds for all states $i$ that $a_i(h) = \sum_{k'=1}^{k-1} b^{k'}(p) + b^{k}(r)$.
By quota, $b^{k'}(p) = p_i$ for all $k$, and $b^k(r) \in \{ \lfloor r \, p_i / p \rfloor, \lceil r \, p_i / p \rceil\}$.
Thus, $a_i(h) \in \{(k-1) \, p_i + \lfloor r \, p_i / p \rfloor, (k-1) \, p_i + \lceil r \, p_i / p \rceil \}$.
The conclusion follows since $h \, p_i /p = ((k-1) \, p + r) \, p_i /p = (k-1) \, p_i + r \, p_i / p$.
\end{proof}

\begin{lemma}
\label{lem:perprob}
For any population profile $\vec{p}$, there is a probability distribution $\mathcal{D}$ over finite seat sequences such that one can sample a finite seat sequence $\alpha \sim \mathcal{D}$ in $\mathcal{O}(p^2 \, n^2)$ randomized time, such that all finite seat sequences in the support of $\mathcal{D}$ satisfy quota, and such that, for all states $i$ and $1 \leq h \leq p$,
\[ \mathbb{P}[\alpha_h = i] = p_i / p. \]
\end{lemma}
\begin{proof}
As sketched in \cref{sec:cumulativeintro}, we define $\mathcal{D}$ by invoking \cref{thm:cumulative} on a star graph with $A = \{a\}, B = \{b_i \mid i \in N\}$, and  $E = \big\{\{a, b_i\} \mid i \in N\big\}$.
We set $T \coloneqq p$, and, for each $1 \leq t \leq T$ and state $i$, set $w_{\{a, b_i\}}^t \coloneqq p_i/p$.

\Cref{thm:cumulative} now defines a joint distribution over variables $X_e^t$ satisfying marginal distribution, degree preservation, and cumulative degree preservation (as well as negative correlation, which we will not use).
We will describe how each joint realization of the $X_e^t$ can be mapped to a finite seat sequence and that the distribution $\mathcal{D}$ that arises from applying this mapping to the dependent-rounding distribution has the properties claim in the statement.
The running time follows from the running time of applying dependent rounding, and the fact that the transformation for translating the outcome into a finite seat sequence requires only $\mathcal{O}(p \, n)$ time.

For a given joint realization of the $X_e^t$, let $\alpha$ be the finite seat sequence that maps each $h \in \{1, \dots, p\}$ to the state $i$ such that $X_{\{a, b_i\}}^h = 1$.
This definition presupposes that there is exactly one such $i$ for each $h$, which follows from the degree-preservation guarantee for vertex $a$ at time step $h$ and from the fact that $d_a^h = \sum_{i \in N} p_i/p = 1$.
This seat sequence $\alpha$ furthermore satisfies quota, which directly follows from cumulative degree preservation and from $\sum_{t'=1}^h d_{b_i}^{t'} = h \, p_i / p$.
It only remains to show that $\alpha_h$ has value $i$ with a probability of $p_i / p = w_{\{a, b_i\}}^h$ for all $i \in N$ and $1 \leq h \leq p$, but this immediately follows from the marginal distribution guarantee of \cref{thm:cumulative}.
\end{proof}

\thmhousemono*
\begin{proof}
For each population profile $\vec{p}$, \cref{lem:perprob} provides a probability distribution over finite seat sequences for that population profile.
We define the outcomes such that $\omega$ contains, for each $\vec{p}$, a separate, independent $\alpha^{\vec{p}}$ following the distribution from \cref{lem:perprob}.

From now on, we fix an $\omega \in \Omega$, which determines the values of all $\alpha^{\vec{p}}$.
For this $\omega$, we must construct an apportionment solution $f = F^{\omega}$.
For a given input $\vec{p}, h$, let $\alpha$ be the concatenation of infinitely many copies of $\alpha^{\vec{p}}$ as in \cref{lem:repetition2}.
Then, we define $f(\vec{p}, h)$ as the apportionment giving $a_i(h) = |\{1 \leq h' \leq h \mid \alpha(h') = i\}|$ seats to each state $i$.

By \cref{lem:perprob}, $\alpha^{\vec{p}}$ satisfies quota, and, thus, $\alpha$ satisfies quota by \cref{lem:repetition2}, from which it follows immediately that $f$ satisfies quota.
Since $f$ was constructed from a seat sequence allocating one seat at a time, it clearly satisfies house monotonicity.

It remains to argue that $F$ satisfies ex ante proportionality.
Fix any $\vec{p}$ and $h$.
By construction, $F_i(\vec{p}, h) - F_i(\vec{p}, h-1) = \bone\big\{\alpha^{\vec{p}}_{1 + (h\!-\!1 \mathop{\mathrm{mod}} p)} = i\big\}$, setting $F_i(\vec{p}, 0) = 0$.
Hence,
\[\mathbb{E}[F_i(\vec{p}, h) - F_i(\vec{p}, h-1)] = \mathbb{P}[\alpha^{\vec{p}}_{1 + (h\!-\!1 \mathop{\mathrm{mod}} p)} = i] =  p_i / p, \]
where the last equality follows from \cref{lem:perprob}. By linearity of expectation, it follows that,
\[ \mathbb{E}[F_i(\vec{p}, h)] = \sum_{h' = 1}^{h} \mathbb{E}[F_i(\vec{p}, h) - F_i(\vec{p}, h-1)] = h \, p_i / p, \]
which shows ex ante proportionality.
\end{proof}

\propquotatone*
\begin{proof}
We will first show a variant of the theorem, in which the finite seat sequences correspond not to perfect matchings but to perfect $b$-matchings, i.e., where each node is labeled with a target degree in $\natszero$, and where a subset of edges is a perfect $b$-matching when each node has its target degree in the induced subgraph.
We will then show how to modify the graph to obtain the claimed result for perfect matchings.

Then, the bipartite graph is the one to which we applied cumulative rounding in \cref{lem:perprob} (without the weights), with the following (technically necessary) modifications:
\begin{enumerate}
    \item We set each node's target degree to its fractional degree as in \cref{lem:fractionaldegrees}. This is possible since all nodes have integer weight, including the nodes $v^{T:T+1}$ which have weight 1 given that $\sum_{t'=1}^T d_v^{t'}$ is an integer for all nodes $v$ in the underlying graph for the chosen $T=p$. 
    \item Then, we delete all edges with zero weight (to ensure that they are never part of the $b$-matching).
    \item Finally, for each edge with weight 1, we delete the edge and decrement the target degree of both adjacent nodes (simulating the constraint that these edges must be present in any $b$-matching).
\end{enumerate}

The proof of \cref{lem:perprob} indicated a way to map perfect $b$-matchings to finite seat sequences, and we have to show that this mapping is a bijection, i.e., that it is injective and surjective.

To show that the mapping is injective, observe that two perfect $b$-matchings that differ in whether a certain edge $\{v^t, (v')^t\}$ is included lead to different seat sequences.
Furthermore, the characterization of the edges in \cref{fig:roundinginterpretation} (whose correctness follows from the proof of \cref{thm:cumulative} and does not rely on properties of the result of dependent rounding other than those required by our $b$-matchings) implies that the set of edges of shape $\{v^t, (v')^t\}$ in the matching uniquely determines which of the other edges are included in the perfect $b$-matching, which means that there are never multiple perfect $b$-matchings that would be mapped to the same finite seat sequence.

It is more involved to show that the mapping is surjective.
For a given finite seat sequence $\alpha$, we will construct a perfect $b$-matching which is mapped to $\alpha$.
Clearly, each edge $\{a^t, (b_i)^t\}$ is included in the matching iff $\alpha_t = i$ (none of these edges have weight zero or one since $n \geq 2$ and each state has a positive population).
We label all other edges according to the edges' events described in \cref{fig:roundinginterpretation}.
One verifies that, by quota, this step would not have taken any edges with zero weight in the cumulative-rounding graph and would have taken all edges with weight one in the cumulative-rounding graph, which allows us to pretend for ease of exposition that we are producing a $b$-matching on the labeled graph before the preprocessing steps (2) and (3).
One easily verifies that the resulting edge set gives the target degree to all nodes of shape $v^t$, $\twobar{v}^t$, and $v^{t:t+1}$ (including the special cases $v^{0:1}$ and $v^{T:T+1}$).

It only remains to show that the nodes $\onebar{v}^t$ have their target degree, $\left\lfloor \sum_{t'=1}^t d_v^{t'} \right\rfloor - \left\lfloor \sum_{t'=1}^{t-1} d_v^{t'} \right\rfloor - \lfloor d_v^t \rfloor + 1$.
For nodes $v=a^t$, it holds that $d_v^{t'} = D_v^{t'} = 1$ for all $t'$, which means that the target degree is one and indeed only one adjacent edge, namely, $\{\onebar{v}^t, v^{t:t+1}\}$, is taken.

We will now consider the case of a node $v=b_i^t$.
Observe that $d_v^{t'} = p_i/\sum_{j\in N} p_j < 1$ for all $t'$, which means that $\sum_{t''=1}^{t'} d_v^{t''} = t' \, p_i / \sum_{j \in N} p_j$, which is just $i$'s standard quota for house size $t'$, which we will write as $q_i(t')$.
Furthermore, note that $\sum_{t'' = 1}^{t'} D_v^{t'} = a_i(t')$.
With this, the target degree of $v$ is just $\qtwo - \qone  + 1$, and the three edges incident to $v$ are selected if
\begin{enumerate}[(a)] %
\item $\atwo = \aone + 1$ (rather than $\atwo = \aone$),
\item $\aone = \qone + 1$ (rather than $\aone = \qone$), and 
\item $\atwo = \qtwo$ (rather than $\atwo = \qtwo$),
\end{enumerate}
respectively, where the values in parentheses are the only alternatives to the properties, by house monotonicity and quota.
That is, we want to show that, for our house monotone and quota $\alpha$, exactly $\qtwo - \qone + 1 \in \{1, 2\}$ many out of the statements (a), (b), and (c), are true.
In \cref{table:casedistinction}, we rule out all other cases via a case distinction, which shows that we indeed produced a perfect $b$-matching, and that the mapping is surjective.
\begin{sidewaystable}
\centering
\scalebox{.9}{%
\begin{tabular}{l l l l p{6cm}}
\toprule
$\qtwo - \qone + 1$ & (a)? & (b)? & (c)? & Impossibility proof\\
\midrule
2 & \patrue & \pbtrue & \pctrue & $\aone = \qone + 1 = \qtwo = \atwo = \aone + 1$ \\ \addlinespace[0.5em]
2 & \pafalse & \pbfalse & ? & $\qtwo = \qone + 1 = \aone + 1 = \atwo + 1 > \atwo$, violates quota \\ \addlinespace[0.5em]
2 & \patrue & \pbfalse & \pcfalse & $\aone = \qone = \qtwo - 1 = \atwo - 2 = \aone - 1$ \\ \addlinespace[0.5em]
2 & \pafalse & \pbtrue & \pcfalse & $\aone = \qone + 1 = \qtwo = \atwo - 1 = \aone - 1$ \\ \addlinespace[0.5em]
1 & \pafalse & \pbfalse & \pcfalse & $\aone = \qone = \qtwo = \atwo - 1 = \aone - 1$ \\\addlinespace[0.5em]
1 & \patrue & ? & \pctrue & $\aone = \atwo - 1 = \qtwo - 1 = \qone - 1$, violates quota \\\addlinespace[0.5em]
1 & \patrue & \pbtrue & \pcfalse & $\aone = \qone + 1 = \qtwo + 1 = \atwo = \aone + 1$ \\\addlinespace[0.5em]
1 & \pafalse & \pbtrue & \pctrue & $\aone = \qone + 1 = \qtwo + 1 = \atwo + 1 = \aone + 1$ \\
\bottomrule
\end{tabular}}
\caption{Case distinction for proving that nodes of the shape $\Bar{b}_i^t$ have the target degree in the proof of \cref{prop:quotatone}.}
\label{table:casedistinction}
\end{sidewaystable}

The above establishes the one-to-one correspondence to the vertices on the polytope of perfect fractional $b$-matchings.
Though this polytope is very nicely behaved already, we prefer to state the theorem for a classical perfect matching polytope, which is more widely known.
Thus, we will adapt the bipartite graph above such that all nodes have target degree one, while keeping the perfect $b$-matchings in one-to-one correspondence.
First, we remove all nodes with target degrees zero from the graph, which clearly does not change the set of perfect $b$-matchings.
Looking at \cref{lem:fractionaldegrees}, only two kinds of nodes can have a target degree larger than one: nodes $a^t$ and some nodes $\onebar{b}_i^t$.
In fact, the nodes $a^t$ are no problem: While they have degree 2 in the cumulative-rounding construction, one of their adjacent edges, $\{a^t, \twobar{a}^t\}$ had weight 1, and thus the target degree of $a^t$ was already lowered to one in step (3) of the preprocessing.

Thus, once more, the only issue are nodes of the form $\onebar{b}_i^t$, specifically, when their target degree is 2 (it is never higher, as discussed above).
If such a node only has two adjacent edges remaining in the graph these edges must be taken in any perfect $b$-matching, so we can eliminate $\onebar{b}_i^t$ and its neighbors.
Thus, say that the node still has all three adjacent edges, $\{\onebar{b}_i^t, \twobar{b}_i^t\}$, $\{\onebar{b}_i^t, b_i^{t-1:t}\}$, and $\{\onebar{b}_i^t, b_i^{t:t+1}\}$.
In this case, replace node $\onebar{b}_i^t$ by two fresh nodes $n_1$ and $n_2$, both with target degree 1, and connect these nodes using four edges $(n_1, \twobar{b}_i^t), (n_1, b_i^{t-1:t}), (n_2, b_i^{t-1:t}), (n_2, b_i^{t:t+1})$.
One verifies that, that in any perfect $b$-matching on the graph before replacement, one can replace the two edges incident to $\onebar{b}_i^t$ by exactly one subset of the new edges to obtain a perfect $b$-matching on the new graph, and that the analogous step in the other direction also works in one unique way.
Thus, after making these replacements, the finite seat sequences satisfying quota correspond one-to-one to the perfect matchings of the graph, which are the corner points of the polytope of fractional perfect matchings by the \bvn{} Theorem.
\end{proof}

\section{Deferred Proofs for Cumulative Rounding}
\label{app:cumulativeproofs}

\lembipartite*
\begin{proof}
Note that the set of nodes
\[ \{ a^t \mid a \in A, 1 \!\leq\! t \!\leq\! T\} \cup \{\onebar{a}^t \mid a \in A, 1 \!\leq\! t \!\leq\! T\} \cup \{\twobar{b}^t \mid b \in B, 1 \!\leq\! t \!\leq\! T\} \cup \{ b^{t:t+1} \mid b \in B, 0 \!\leq\! t \!\leq\! T \}\]
has no internal edges, and neither does the complement of this set.
\end{proof}

\lemweightunit*
\begin{proof}
If the edge has the shape $\{v^t, (v')^t\}$ for some $v, v' \in A \cup B$, then the edge weight is one of the $w_e^t$, which are in $[0,1]$ by assumption.
All other edge weights either have the shape $x - \lfloor x \rfloor$ or the shape $1 - x + \lfloor x \rfloor = 1 - (x - \lfloor x \rfloor)$ for some $x \in \mathbb{R}$.
The claim follows since $x - 1 < \lfloor x \rfloor \leq x$.
\end{proof}

\lemfractionaldegrees*
\begin{proof}
Within this proof, denote by $\frdeg(\cdot)$ the fractional degree of a vertex in the constructed graph.
\begin{align*}
 \frdeg(v^t) &= \underbrace{\textstyle{\sum_{v \in e \in E} w_e^t}}_{=d_v^t}
+ \left(1 - d_v^t + \lfloor d_v^t \rfloor\right) = \lfloor d_v^t \rfloor + 1 \\
    \frdeg(\onebar{v}^t) &= \left( d_v^t - \lfloor d_v^t \rfloor \right)
    + \left( \textstyle\sum_{t'=1}^{t-1} d_v^{t'} - \lfloor \textstyle\sum_{t'=1}^{t-1} d_v^{t'}\rfloor \right)
    + \left( 1 - \textstyle\sum_{t'=1}^t d_v^{t'} + \lfloor \textstyle\sum_{t'=1}^t d_v^{t'} \rfloor \right) \\
    &= \lfloor \textstyle\sum_{t'=1}^t d_v^{t'} \rfloor - \lfloor d_v^t \rfloor 
     - \lfloor \textstyle\sum_{t'=1}^{t-1} d_v^{t'}\rfloor 
    +  1 \\
    \frdeg(\twobar{v}^t) &= \left( 1 - d_v^t + \lfloor d_v^t \right) + \left( d_v^t - \lfloor d_v^t \rfloor \right) = 1 \\
    \frdeg(v^{t:t+1}) &= \left( 1 - \textstyle\sum_{t'=1}^t d_v^{t'} + \lfloor \textstyle\sum_{t'=1}^t d_v^{t'} \rfloor \right) + \left( \textstyle\sum_{t'=1}^t d_v^{t'} - \lfloor \textstyle\sum_{t'=1}^t d_v^{t'} \rfloor \right) = 1  \quad\text{(if $1 \leq t \leq T-1$)} \\
    \frdeg(v^{0:1}) &= \textstyle\sum_{t'=1}^0 d_v^{t'} - \lfloor \textstyle\sum_{t'=1}^0 d_v^{t'} \rfloor = 0 - \lfloor 0 \rfloor = 0 \qedhere
\end{align*}

\end{proof}

\end{document}